\documentclass[11pt]{article}
\usepackage{amsfonts}
\usepackage{amsmath,amsthm}

\usepackage{mathrsfs}
\usepackage{cite}
\usepackage{color}
\usepackage{hyperref}
\usepackage{graphicx,subfigure}
\usepackage[numbers,sort&compress]{natbib}
\usepackage[top=1 in,bottom=1 in,left=1.0 in,right=1.0 in]{geometry}

\newcommand{\bblu}{\begin{color}{blue}}
\newcommand{\bred}{\begin{color}{red}}
\newcommand{\ecl}{\end{color}}
\numberwithin{equation}{section}

\newtheorem{proposition}{Proposition}[section]
\newtheorem{lemma}{Lemma}[section]
\newtheorem{theorem}{Theorem}[section]
\newtheorem{remark}{Remark}[section]

\def\t#1{\widetilde{#1}}
\def\h#1{\widehat{#1}}

\def\b#1{\overline{#1}}
\def \tb#1{\overline{\widetilde{#1}}}

\begin{document}

\title{Algebro-geometric solutions to the lattice potential modified
Kadomtsev--Petviashvili equation}
\author{Xiaoxue Xu$^{1}$, ~Cewen Cao$^{1}$, ~Da-jun Zhang$^{2}$
\\
{\small ${}^{1}$School of Mathematics and Statistics, Zhengzhou University, Zhengzhou, Henan 450001, P.R. China}\\
{\small ${}^{2}$Department of Mathematics, Shanghai University, Shanghai 200444, P.R. China }\\
{\small E-mail:\ \ xiaoxuexu@zzu.edu.cn,~ cwcao@zzu.edu.cn,~ djzhang@staff.shu.edu.cn}
}

\date{\today}

\maketitle

\begin{abstract}
Algebro-geometric solutions of the  lattice potential modified Kadomtsev--Petviashvili (lpmKP) equation are constructed.
A Darboux transformation of the Kaup--Newell spectral problem is employed to generate a Lax triad for the
lpmKP equation, as well as to define commutative integrable symplectic maps
which generate discrete flows of eigenfunctions.
These maps share the same integrals with the finite-dimensional Hamiltonian system
associated to the Kaup--Newell spectral problem.
We investigate asymptotic behaviors of the Baker--Akhiezer functions and obtain their expression
in terms of Riemann theta function.
Finally, algebro-geometric solutions for the lpmKP equation are reconstructed from these Baker--Akhiezer functions.


\end{abstract}

\section{Introduction }\label{sec-1}

In recent two decades discrete integrable systems have undergone a true development
(see \cite{Hietarinta} and the references therein).
One of remarkable results is the Adler--Bobenko--Suris (ABS) classification of quadrilateral equations
that are consistent around a cube (CAC) \cite{NW-2001,Nij-2002,BS-2002,ABS}
with certain extra restrictions (affine linear, D4 symmetry and tetrahedron property) \cite{ABS}.
The ABS list contains only 9 equations, named H$_1$, H$_2$, H$_3(\delta)$, A$_1(\delta)$,
A$_2$, Q$_1(\delta)$, Q$_2$, Q$_3(\delta)$ and Q$_4$.
The ABS equations serve as main objects in the study of discrete integrable systems,
although many equations in the list are already found before.
During the study, some methods have been developed or elaborated for discrete integrable systems,
such as the Cauchy matrix approach \cite{NAH-2009,ZZ-2013}, bilinear method \cite{HZ-2009},
inverse scattering transform \cite{But-Non-2012} and algebraic geometry approach \cite{CaoX-JPA-2012}.

The finite-gap integration method was created for solving the
Korteweg--de Vries (KdV) equation with periodic initial value problem
by Novikov, Maveev and their collaborators Dubrovin, Its and Krichever in 1970s
\cite{Dubrovin-1975,DN-1975,IM-1975a,IM-1975b,Krichever-1977, Novikov-1974}.
The obtained periodic solutions are called finite-gap solutions or algebro-geometry solutions.
After the original work, the theory has undergone a true development
(e.g.\cite{BBEIM-1994,GH-2003,Gesztesy}) and had a strong impact on the evolution of modern mathematics
and theoretical physics (see \cite{Matveev} and the references therein
for the historical review of finite-gap integration method).

For fully discrete equations, such as the ABS equations,
algebro-geometry solutions have been derived from several equations in a series of papers
\cite{CaoX-JPA-2012,CXZ,CZ2012a,CZ2012b,XCN,XCZ-JNMP-2020,XJN} by Cao, Xu and their collaborators,
which establish an algebro-geometric approach to periodic solutions
for discrete equations that are multi-dimensionally consistent.
This approach can be sketched as the following three steps.
For a given quadrilateral equation which admits a Lax pair, in the first step we find
an associated continuous spectral problem that is compatible with the discrete Lax pair.
In step 2, consider the Hamiltonian system that arises from the associated continuous spectral problem
and derive independent integrals of the  Hamiltonian system.
Then, view the discrete Lax pair as maps by which discrete flows of eigenfunctions can be generated.
Using the compatibility between the maps and the associated continuous spectral problem,
one can prove that the maps are symplectic and integrable in Liouville sense, sharing
the same integrals with the Hamiltonian system.
Finally, in step 3, using algebro-geometric technique, introduce  Baker--Akhiezer functions and
Abel--Jacobi variables,  and express the Baker--Akhiezer functions in terms of the Riemann
theta function according to their divisors,
and recover the discrete potential functions which are expressed in terms of the Riemann theta function.
The approach was also extended to the lattice potential Kadomtsev--Petviashvili (lpKP) equation \cite{CXZ}
which is 4D consistent  \cite{ABS-2012}.

The first step of the approach is a starting point but it is highly nontrivial.
In fact, since the Lax pair are compatible with the  associated continuous spectral problem,
the Lax pair can be considered as Darboux transformations of the continuous spectral problem
(cf.\cite{Levi-1980}).
With regard to an ABS equation, its spectral problem is thought to be
a Darboux transformation of some continuous spectral problem,
but such associated continuous spectral problems are still unknown for
H$_2$, H$_3(\delta)$, Q$_2$, Q$_3(\delta)$ and Q$_4$.
In this step, one needs to secure a continuous spectral problem and simultaneously
to find a ``suitable'' Darboux transformation which not only can act as a spectral problem
of the target discrete equation, but also can be used to recover the discrete potential functions
in the follow-up steps.
And moreover, the Darboux transformation is usually different from the discrete spectral problem obtained from
multidimensional consistency.
Of course, the step 2 and 3 are also different for different equations.
Let us list out the associated continuous spectral problems that have been used in this approach.
The Schr\"odinger spectral problem (in matrix form)
\begin{equation*}
\phi_x=\left(\begin{array}{cc}
                  0 & -\lambda +u\\
                  1 & 0
                  \end{array}\right)\phi
\end{equation*}
was used in \cite{CaoX-JPA-2012} for solving H$_1$.
To solve Hirota's discrete sine-Gordon equation, spectral problems
\begin{equation*}
\phi_x=\left(\begin{array}{cc}
                  u &  \lambda \\
                  \lambda & -u
                  \end{array}\right)\phi
\end{equation*}
was employed \cite{CZ2012a}.
Ref.\cite{XCZ-JNMP-2020} made use of
\begin{equation*}
\phi_x=\left(\begin{array}{cc}
                  u & \lambda\\
                  0 & -u
                  \end{array}\right)\phi
\end{equation*}
to solve Q$_1(0)$.
Ref.\cite{XJN} solves H$_1$, H$_3(0)$ and Q$_1(0)$ and the three
associated continuous spectral problems are, respectively,
\begin{equation*}
\phi_x=\left(\!\begin{array}{cc}
                  v &  -\lambda+u \\
                  1 & -v
                  \end{array}\!\right)\phi,~~
\phi_x=\left(\begin{array}{cc}
                  \frac{\lambda^2}{2} &  \lambda u \\
                  \lambda v  & -\frac{\lambda^2}{2}
                  \end{array}\!\right)\phi,~~
\phi_x=\left(\!\!\begin{array}{cc}
                 -\frac{\lambda^2}{2}+u+v &  \lambda u \\
                  -\lambda   & \frac{\lambda^2}{2}-u-v
                  \end{array}\!\!\right)\phi,
\end{equation*}
where the second one is known as the Kaup--Newell spectral problem.
In \cite{XCN}, to solve Q$_1(\delta)$,  spectral problem
\begin{equation*}
\phi_x=\frac{1}{\lambda u_x}\left(\begin{array}{cc}
                  0 & \lambda^2\delta^2+u_x^2\\
                  1 & 0
                  \end{array}\right)\phi
\end{equation*}
has been used. For the coupled lattice nonlinear Schr\"odinger equations, a nonsymmetric 3D lattice,
and the  lpKP equation, the
Zakharov--Shabat--Ablowitz--Kaup--Newell--Segur (ZS--AKNS) spectral problem
\begin{equation}\label{AKNS-sp}
\phi_x=\left(\begin{array}{cc}
                  -\lambda & u\\
                  v & \lambda
                  \end{array}\right)\phi
\end{equation}
was employed respectively in \cite{CZ2012b} and \cite{CXZ}.
Note that one discrete equation may have more than one associated continuous spectral problems,
e.g. H$_1$ and Q$_1(0)$.

This paper is devoted to finding algebro-geometric solutions
to the lattice potential modified KP (lpmKP) equation,
\begin{equation}\label{eq:1.1}
\Xi^{(0,3)} \equiv
\beta_1^2(e^{-\overline{\widetilde{W}}+\overline{W}}-e^{-\widehat{\widetilde{W}}+\widehat{W}})
+\beta_2^2(e^{-\widehat{\overline{W}}+\widehat{W}}-e^{-\widetilde{\overline{W}}+\widetilde{W}})
+\beta_3^2(e^{-\widetilde{\widehat{W}}+\widetilde{W}}-e^{-\overline{\widehat{W}}+\overline{W}})=0,
\end{equation}
where $j, k$ in the notation  $\Xi^{(j,k)}$ respectively stand for the numbers of continuous and discrete
independent variables in the equation,
$\beta_j (j=1,2,3)$ are the lattice parameters associated with three directions $m_j (j=1,2,3)$,
and $\widetilde{W}, \overline{W}, \widehat{W}$ are conventional notations denoting shifts in different directions, i.e.
\begin{equation*}
\widetilde{W}=W(m_1+1,m_2,m_3),~~ \overline{W}=W(m_1,m_2+1,m_3),~~
\widehat{W}=W(m_1,m_2,m_3+1).
\end{equation*}
This equation is first given in \cite{NCW-LNP-1985} (Eq.(4.16)),
derived in a framework of direct linearisation for 3D lattice equations developed in \cite{Nijhoff}.
It is one of the five octahedron-type equations that are 4D consistent \cite{ABS-2012}.
In continuous limit the equation goes to the potential mKP equation \cite{NCW-LNP-1985}
\begin{equation}\label{eq:1.2}
\Xi^{(3,0)}\equiv\frac{1}{4}(W_{xxx}-2W_x^3)_x-\frac{3}{2}W_{xx}W_y+\frac{3}{4}W_{yy}-W_{xt}=0.
\end{equation}
The associated continuous spectral problem we will employ to solve the lpmKP equation is the KN spectral problem \cite{KaupN-1977,Kaup}
\begin{equation}\label{eq:1.3}
\partial_x\chi=U_1\chi,\quad U_1(\lambda;u,v)=\left(\begin{array}{cc}
\lambda^2/2&\lambda u\\\lambda v&-\lambda^2/2
\end{array}\right).
\end{equation}

The paper is organized as follows.
In Section \ref{sec-2} we present a Lax triad for the lpmKP equation,
explain connections between the Lax triad and the associated KN spectral problem.
In Section \ref{sec-3}, a nonlinear integrable symplectic map is obtained
based on a finite-dimensional Hamiltonian system arising from the KN spectral problem.
In Section \ref{sec-4}, finite-gap solutions to the lpmKP equation  \eqref{eq:1.1} are constructed
by introducing the Baker-Akhiezer functions and
using a discrete analogue of the Liouville-Arnold theory.
After that, we present an example of the obtained solutions for the case of genus one in Section \ref{sec-5}.
The concluding remarks are given in Section \ref{sec-6}.
We will also obtain a hierarchy of equations that are related to the lpmKP equation,
in terms of the number of discrete independent variables.
Together with  some links between $(1+1)$-dimensional and $(2+1)$-dimensional integrable systems,
these will be given in Appendix \ref{A-1}.
In addition, Appendix \ref{B-1}
provides some complex algebraic geometry preliminaries that will be used in our approach.

\section{The lpmKP equation arising from compatibility}\label{sec-2}

The idea of introducing discretisation by transform originated in \cite{Levi-JPA-1981,Levi-1980}.
A Darboux transformation of a continuous spectral problem can serve as a discrete spectral problem
to generate semi-discrete and fully discrete integrable systems,
e.g. \cite{CZ-CPL-2012,KMW-TMP-2013,KMX-JMP-2015,Nimmo,ZhouC-CTP-2015}.
In this section, we will explain how the lpmKP equation \eqref{eq:1.1}
is connected with the KN spectral problem \eqref{eq:1.3}.
These connections are important for us to reconstruct the discrete potential function $W$ in the follow-up sections.

Consider a Darboux transformation of the KN spectral problem \eqref{eq:1.3}: (cf.\cite{ZhouC-CTP-2015})
\begin{equation}\label{eq:1.4}
T_1\chi\equiv\t{\chi}=D^{(\beta)}\chi,\quad
D^{(\beta)}=\left(\begin{array}{cc}
                    \lambda^2(ab+1)-\beta^2&\lambda a \\
                    \lambda b&1
                    \end{array}\right),
\end{equation}
where $\beta$ serves as a soliton parameter,
which transforms \eqref{eq:1.3} to
$\partial_x \t\chi =\t U_1 \t\chi=U_1(\lambda, \t u, \t v)\t \chi$.
This requires that\footnote{Considering \eqref{eq:1.4} as a discrete spectral problem,
equation \eqref{eq:2.1} is the semi-discrete zero curvature equation from the compatibility
between \eqref{eq:1.4} and \eqref{eq:1.3}}
\begin{equation}\begin{split}\label{eq:2.1}
0&=D^{(\beta)}_x-\widetilde{U}_1 D^{(\beta)}+D^{(\beta)}U_1\\
&=\left(\begin{array}{cc}
\lambda^2(Z_x-\tilde{u}b+va)&\lambda^3(-a+Zu)+\lambda(a_x-\tilde{u}-\beta^2u)\\
\lambda^3(b-Z\tilde{v})+\lambda(b_x+v+\beta^2\tilde{v})&\lambda^2(ub-\tilde{v}a)
\end{array}\right),
\end{split}\end{equation}
where we have taken
\begin{subequations}\label{eq:2.2}
\begin{equation}\label{eq:2.2a0}
Z=ab+1.
\end{equation}
It follows that
\begin{align}
&a_x-\t{u}-\beta^2u=0,\quad b_x+v+\beta^2\t{v}=0,\label{eq:2.2a}\\
&a=Zu,\quad b=Z\t{v},\label{eq:2.2b}\\
&Z_x=(\t{u}\t{v}-uv)Z,\label{eq:2.2c}
\end{align}\end{subequations}
which allows a special formulation\footnote{Note that this formulation
is different from the Case $T_{22}$ in \cite{ZhouC-CTP-2015} and therefore the spectral problem \eqref{eq:1.4}
(with formulation \eqref{eq:2.3.4}) is different from the spectral problem introduced in \cite{WG-JPA-1998}.}
\begin{subequations}\label{eq:2.3.4}
\begin{align}
&Z=2/(1+\sqrt{1-4u\tilde{v}}),\label{eq:2.3}\\
&a=2u/(1+\sqrt{1-4u\tilde{v}}),\quad
b=2\tilde{v}/(1+\sqrt{1-4u\tilde{v}}),\label{eq:2.4}
\end{align}
\end{subequations}
As a byproduct, equation (\ref{eq:2.2a}) (in terms of the variables $(u,v)$) provides
a B\"{a}cklund transformation for the KN potentials in \eqref{eq:1.3},
\begin{subequations}\label{eq:2.5}
\begin{align}
&\Xi_1^{(1,1)}\equiv u_x+(\tilde{u}\tilde{v}-uv)u-\frac{1}{2}(1+\sqrt{1-4u\tilde{v}})(\tilde{u}+\beta^2u)=0,
\label{eq:2.5a}\\
&\Xi_2^{(1,1)}\equiv \tilde{v}_x+(\tilde{u}\tilde{v}-uv)\tilde{v}
+\frac{1}{2}(1+\sqrt{1-4u\tilde{v}})(v+\beta^2\tilde{v})=0,\label{eq:2.5b}
\end{align}
\end{subequations}
which can also be regarded as a semi-discrete derivative nonlinear Schr\"odinger (dNLS) equation
as it continuum limit yields the dNLS equation \eqref{eq:2.17a} (see Appendix \ref{A-1}).

\begin{remark}\label{R-2.1}
Note that $Z$ given in (\ref{eq:2.3}) is one of roots of the  quadratic equation $Z=Z^2u\t{v}+1$,
which results from the relation \eqref{eq:2.2}.
If we take another root, $Z=2/(1-\sqrt{1-4u\tilde{v}})$,
we will have an analogue of \eqref{eq:2.5}:
\begin{align*}
&\Xi_1^{'(1,1)}\equiv u_x+(\tilde{u}\tilde{v}-uv)u-\frac{1}{2}(1-\sqrt{1-4u\tilde{v}})(\tilde{u}+\beta^2u)=0,
\\
&\Xi_2^{'(1,1)}\equiv \tilde{v}_x+(\tilde{u}\tilde{v}-uv)\tilde{v}
+\frac{1}{2}(1-\sqrt{1-4u\tilde{v}})(v+\beta^2\tilde{v})=0.
\end{align*}
However, it cannot recover the  continuous dNLS equation \eqref{eq:2.17a} in continuum limit
due to the fatal minors sign in front of the square roots.
In this context, in the paper we will only consider the consequential results of \eqref{eq:2.3}.
\end{remark}

Next, in order to derive the lpmKP equation (\ref{eq:1.1}),
we consider discrete spectral problems, which are three replicas of equation (\ref{eq:1.4})
with the formulation \eqref{eq:2.3.4},
\begin{subequations}\label{eq:2.6}
\begin{align}
&T_1\chi=D^{(\beta_1)}\chi=\left(\begin{array}{cc}
           \lambda^2 Z^{(1)}-\beta_1^2&\lambda Z^{(1)} u\\
           \lambda Z^{(1)}\t{v}&1\end{array}\right)\chi,\label{eq:2.6a} \\
&T_2\chi=D^{(\beta_2)}\chi=\left(\begin{array}{cc}
           \lambda^2 Z^{(2)}-\beta_2^2&\lambda Z^{(2)} u\\
           \lambda Z^{(2)} \b{v}&1\end{array}\right)\chi,\label{eq:2.6b} \\
&T_3\chi=D^{(\beta_3)}\chi=\left(\begin{array}{cc}
           \lambda^2 Z^{(3)}-\beta_3^2&\lambda Z^{(3)} u\\
           \lambda Z^{(3)} \h{v}&1\end{array}\right)\chi, \label{eq:2.6c}
\end{align}
\end{subequations}
where $T_i$ stands for the shift operator in $m_i$-direction, i.e.
$T_1 f=\t f, T_2 f=\b f, T_3 f=\h f$, $\beta_i$ serves as the spacing parameter of the $m_i$-direction,
and in light of \eqref{eq:2.3},
\begin{equation}\label{eq:2.7}
Z^{(1)}=\frac{2}{1+\sqrt{1-4u\t{v}}},\quad
Z^{(2)}=\frac{2}{1+\sqrt{1-4u\b{v}}},\quad
Z^{(3)}=\frac{2}{1+\sqrt{1-4u\h{v}}}.
\end{equation}
One can quickly check that the compatibility between \eqref{eq:2.6a} and \eqref{eq:2.6b}
requires that
\begin{equation}\label{eq:ZZ}
\b Z^{(1)} Z^{(2)}=\t Z^{(2)} Z^{(1)}.
\end{equation}
This relation, together with \eqref{eq:2.2c}, leads us to introducing
\begin{equation}\label{eq:2.8d}
W_x=uv
\end{equation}
and
\begin{subequations}\label{eq:2.8}
\begin{align}
&\widetilde{W}-W=\ln Z^{(1)},\label{eq:2.8a}\\
&\overline{W}-W=\ln Z^{(2)},\label{eq:2.8b}\\
&\widehat{W}-W=\ln Z^{(3)}.\label{eq:2.8c}
\end{align}
\end{subequations}
Later in Section \ref{sec-4} we will derive $Z^{(i)}$ in an explicit form,
from which $W$ can be ``integrated'' from \eqref{eq:2.8}. In this sense,
$u$ and $v$ act as auxiliary variables and $x$ is a dummy variable.
Note that explicit $u$ and $v$ can also be obtained from $Z^{(i)}$ (see Remark \ref{R-4.1}).

In order to show that $W$ defined by the triad \eqref{eq:2.6} via \eqref{eq:2.8}
satisfies the lpmKP equation \eqref{eq:1.1},
we first look at equations resulting from the compatibility of  the first two equations in \eqref{eq:2.6}.

\begin{lemma}\label{L-2-1}
If $u$ and $v$ are functions of $(m_1,m_2)$ such that
$W(m_1,m_2)$ can be well defined by equations (\ref{eq:2.8a}) and (\ref{eq:2.8b}),
and  equations (\ref{eq:2.6a}) and (\ref{eq:2.6b}) are compatible for $\chi$,
then the pair $(u,v)$ solves the lattice dNLS (ldNLS) equation
\begin{subequations}\label{eq:2.9}
\begin{align}
&\Xi_1^{(0,2)}\equiv \t{Z}^{(2)} \t{u}-\b Z^{(1)}\b{u}
+(\beta_1^2 Z^{(2)}-\beta_2^2 Z^{(1)})u=0,\label{eq:2.9a}\\
&\Xi_2^{(0,2)}\equiv Z^{(1)}\t{v}-Z^{(2)}\b{v}
-(\beta_1^2 \t{Z}^{(2)}-\beta_2^2 \b Z^{(1)})\t{\b{v}}=0,\label{eq:2.9b}
\end{align}
\end{subequations}
where $Z^{(j)}$ are defined in \eqref{eq:2.7}.
Moreover, the following relation
\begin{equation}\label{eq:2.10}
Y_{12}\equiv\beta_1^2 \Bigl[\Bigl(\b{Z}^{(1)}\Bigr)^{-1}-\Bigl(Z^{(1)}\Bigr)^{-1}\Bigr]
-\beta_2^2 \Bigl[\Bigl(\t{Z}^{(2)}\Bigr)^{-1}-\Bigl(Z^{(2)}\Bigr)^{-1}\Bigr]
+(\t{u}\t{v}-\b{u}\,\b{v})=0
\end{equation}
holds.
\end{lemma}

\begin{proof}
By making use of relation \eqref{eq:ZZ}, the compatibility of equations (\ref{eq:2.6a}) and (\ref{eq:2.6b}) gives rise to
\begin{equation*}
\mathbf{0}=\widetilde{D}^{(\beta_2)}D^{(\beta_1)}-\overline{D}^{(\beta_1)}D^{(\beta_2)}
=\left(\begin{array}{cc}
\lambda^2e^{\tb W-W}Y_{12} & \lambda\Xi_1^{(0,2)}\\
\lambda \Xi_2^{(0,2)} & 0
\end{array}\right),
\end{equation*}
where $\Xi_1^{(0,2)}$ and $\Xi_2^{(0,2)}$ are given in \eqref{eq:2.9}.
In addition, noticing that
\begin{equation}\label{eq:2.12}
Y_{12}=e^{W-\t{\b{W}} }\Bigl(Z^{(1)}\t{v}\Xi_1^{(0,2)}+\overline{Z}^{(1)}\b{u}\Xi_2^{(0,2)}\Bigr),
\end{equation}
we obtain equation \eqref{eq:2.10} as a consequence of \eqref{eq:2.9}.
The proof is then completed.

\end{proof}

Note that in light of (\ref{eq:2.7}), equations (\ref{eq:2.9a}) and (\ref{eq:2.9b}) can be written as
\begin{subequations}
\begin{align}
&\frac{1}{2}(1+\sqrt{1-4u\t{v}}\,)(\t{u}+\beta_1^2u)
-\frac{1}{2}(1+\sqrt{1-4u\b{v}}\,) (\b{u}+\beta_2^2u)-(\t{u}\t{v}-\b{u}\,\b{v})u=0, \\
&\frac{1}{2}(1+\sqrt{1-4\t{u}\b{\t{v}}}\,)(\t{v}+\beta_2^2\b{\t{v}})
-\frac{1}{2}(1+\sqrt{1-4\b{u}\b{\t{v}}}\,)(\b{v}+\beta_1^2\b{\t{v}})-(\t{u}\t{v}-\b{u}\,\b{v})\b{\t{v}}=0.
\end{align}
\end{subequations}

Next, we can show that $W$ satisfies the lpmKP equation \eqref{eq:1.1}
after consistently introducing evolution in the third direction.

\begin{theorem}\label{T-2-1}
If $u$ and $v$ are functions of $(m_1,m_2,m_3)$ such that
$W(m_1,m_2,m_3)$ can be well defined by  (\ref{eq:2.8}) and
equations in the  triad  (\ref{eq:2.6}) are compatible for $\chi$,
then  $W$ satisfies the lpmKP equation \eqref{eq:1.1}.
\end{theorem}

\begin{proof}
In fact, in addition to the relation \eqref{eq:2.12},
compatibility between any two equations in the triad \eqref{eq:2.6} yields
\begin{align*}
Y_{23} & \equiv \beta_2^2 \Bigl[\Bigl(\h{Z}^{(2)}\Bigr)^{-1}-\Bigl(Z^{(2)}\Bigr)^{-1}\Bigr]
-\beta_3^2 \Bigl[\Bigl(\b{Z}^{(3)}\Bigr)^{-1}-\Bigl(Z^{(3)}\Bigr)^{-1}\Bigr]
+(\b{u}\,\b{v}-\h{u}\h{v})=0,\\
Y_{31} & \equiv \beta_3^2 \Bigl[\Bigl(\t{Z}^{(3)}\Bigr)^{-1}-\Bigl(Z^{(3)}\Bigr)^{-1}\Bigr]
-\beta_1^2 \Bigl[\Bigl(\h{Z}^{(1)}\Bigr)^{-1}-\Bigl(Z^{(1)}\Bigr)^{-1}\Bigr]
+(\h{u}\,\h{v}-\t{u}\t{v})=0.
\end{align*}
Then, noticing that
\begin{equation}\label{eq:Y123}
\Xi^{(0,3)}=Y_{12}+Y_{23}+Y_{31}
\end{equation}
and replacing $Z^{(i)}$ by $W$ using \eqref{eq:2.8},
we arrive at the lpmKP equation \eqref{eq:1.1}.
The proof is completed.

\end{proof}

Since the spectral problem \eqref{eq:2.6a} arises from the Darboux transformation \eqref{eq:1.4},
which commutes with the KN spectral problem and therefore may serve as a
Darboux transformation for the whole KN hierarchy,
we can have more equations in this frame,
which have different number of discrete independent variables and
compose a hierarchy of the lpmKP family.
Since we will focus on the lpmKP equation \eqref{eq:1.1},
we will list out these semi-discrete equations in Appendix \ref{A-1}.

\section{ Nonlinear integrable map $S_\beta$}\label{sec-3}

In this section, we will prove that the Darboux transformation \eqref{eq:1.4}
(with \eqref{eq:2.2}) is an integrable symplectic map.

For the sake of self-containedness of the paper, let us recall some results given in \cite{CY}
for the Hamiltonian system associated with the KN spectral problem \eqref{eq:1.3}.
Let $N$ be any positive integer,  $<\xi,\eta>=\Sigma_{j=1}^N\xi_j\eta_j$,
and $A\doteq{\rm diag}(\alpha_1,\cdots,\alpha_N)$ with distinct and non-zero $\alpha_1^2,\cdots,\alpha_N^2$.
A Liouville integrable Hamiltonian system $(\mathbb{R}^{2N},{\rm d}p\wedge{\rm d}q, H_1)$
is constructed by $N$ copies of the KN spectral problem
(\ref{eq:1.3}), by imposing a constraint \eqref{eq:3.1c} on  $(u,v)$,  as
\begin{subequations}\label{eq:3.1}
\begin{align}
&H_1=-\frac{1}{2}<A^2p,q>+\frac{1}{2}<Ap,p><Aq,q>,\label{eq:3.1a}\\
&\partial_x{p_j\choose q_j}={-\partial H_1/\partial q_j\choose\partial H_1/\partial p_j}
=U_1(\alpha_j;u,v){p_j\choose q_j},\quad 1\le j\le N,\label{eq:3.1b}\\
&u=-<Ap,p>,\quad v=<Aq,q>, \label{eq:3.1c}
\end{align}
\end{subequations}
where $p=(p_1, p_2, \cdots, p_N)^T$ and $q=(q_1, q_2, \cdots, q_N)^T$.
Note that equation \eqref{eq:3.1c}, coinciding with squared eigenfunction symmetry constraint,
converts the KN spectral problem \eqref{eq:1.3} to the Hamiltonian equation \eqref{eq:3.1b},
which is nonlinear with respect to the eigenfunction $(p,q)$.
Such a procedure is usually referred to as \textit{nonlinearisation} of a Lax pair (cf.\cite{Cao1990}).
The associated Lax equation $L_x=[U_1,L]$ has a solution, which is
\begin{equation}\label{eq:3.2}
L(\lambda;p,q)=\left(\begin{array}{cc}
L^{11}(\lambda) & L^{12}(\lambda) \\
L^{21}(\lambda) &-L^{11}(\lambda)
\end{array}\right)
=\left(\begin{array}{cc}
\frac{1}{2}+Q_\lambda(A^2p,q)&-\lambda Q_\lambda(Ap,p)\\
\lambda Q_\lambda(Aq,q) &-\frac{1}{2}-Q_\lambda(A^2p,q)
\end{array}\right),
\end{equation}
where $Q_\lambda(\xi,\eta)=<(\lambda^2-A^2)^{-1}\xi,\eta>$.
It satisfies the $r$-matrix Ansatz \cite{Babelon,Faddeev,Gerdjikov}
\begin{equation}\label{eq:3.3}
\{L(\lambda) \underset{,}\otimes L(\mu)\}=[r(\lambda,\mu),L(\lambda)\otimes
I]+[r^\prime(\lambda,\mu),I\otimes L(\lambda)],
\end{equation}
where
\begin{align*}
& r(\lambda,\mu)=\frac{2\lambda}{\lambda^2-\mu^2}P_{\lambda\mu},\quad
r^\prime(\lambda,\mu)=\frac{2\mu}{\lambda^2-\mu^2}P_{\mu\lambda}=-r(\mu,\lambda), \\
& P_{\lambda\mu}=\left(\begin{array}{cccc}
                             \lambda&0&0&0\\
                             0&0&\mu&0\\
                             0&\mu&0&0\\
                             0&0&0&\lambda
                             \end{array}\right).
\end{align*}
This implies the Poisson commutativity $\{F(\lambda),F(\mu)\}=0$,
where $F(\lambda)=\det L(\lambda)$ (cf.\cite{XCN})
and the Poisson bracket is defined as
\[\{A,B\}=\sum^{N}_{k=1}\biggl(\frac{\partial A}{\partial q_k}\frac{\partial B}{\partial p_k}-
\frac{\partial A}{\partial p_k}\frac{\partial B}{\partial q_k}\biggr).
\]
The Hamiltonian $H_1$ can be determined by the expansion
\begin{equation}\label{eq:3.4}
H(\lambda)=-\displaystyle\frac{\sqrt{-F(\lambda)}}{2}=-\displaystyle\frac{1}{4}
+\sum_{j=1}^\infty H_j\zeta^{-j},\ \  \zeta=\lambda^2,
\end{equation}
which implies that $\{H_1,F(\lambda)\}=0$.
The generating function $F(\lambda)$ can be expanded as
\begin{equation}\label{eq:3.5}
F(\lambda)=-\frac{1}{4}+\sum_{k=1}^N\frac{\alpha_k^2E_k}{\lambda^2-\alpha_k^2},
\end{equation}
which yields a complete set of integrals for the Hamiltonian system
$(\mathbb{R}^{2N},{\rm d}p\wedge{\rm d}q, H_1)$,
\begin{equation}\label{eq:3.6}
E_k= (2<p,q>-1)p_k q_k-p_k^2 q_k^2
 +\frac{\alpha_k}{2}\sum_{\begin{subarray}{c}1\leq j \leq N;\\ j\neq k \end{subarray}}
\Big(\frac{(p_jq_k-p_kq_j)^2}{\alpha_k-\alpha_j}-\frac{(p_jq_k+p_kq_j)^2}{\alpha_k+\alpha_j}\Big),
\end{equation}
where $1\leq k\leq N$.
Suppose that the roots of $F(\lambda)$ are $\zeta_j=\lambda_j^2, ~j=1,\ldots, N$,
then we have the factorization
\begin{equation}\label{eq:3.7}
F(\lambda)=-\frac{\prod_{j=1}^N(\zeta-\lambda_j^2)}{4\alpha(\zeta)}=-\frac{R(\zeta)}{4\alpha^2(\zeta)},
\end{equation}
where $\alpha(\zeta)=\prod_{j=1}^N(\zeta-\alpha_j^2), R(\zeta)=\alpha(\zeta)\prod_{j=1}^N(\zeta-\lambda_j^2)$.
Thus a hyperelliptic curve
\begin{equation}\label{R-curve}
\mathcal R:~~ \xi^2=R(\zeta)=\prod_{j=1}^N(\zeta-\alpha_j^2)(\zeta-\lambda_j^2),
\end{equation}
with genus $g=N-1$, is defined.
The Riemann surface where $\zeta$ is consists of two sheets,
and the curve $\mathcal{R}$ is of hyperelliptic involution
in the sense that $\tau: (\zeta, \xi) \to (\zeta, -\xi)$ maps $\mathcal{R}$ to itself.
For a non-branching point $\zeta$ on the Riemann surface,
when necessary, we distinguish the two corresponding points on $\mathcal R$ by
\begin{equation}\label{eq:3.8}
\mathfrak p_+(\zeta)=\big(\zeta,\,\xi=\sqrt{R(\zeta)}\big),\quad ~
\mathfrak p_-(\zeta)=\big(\zeta,\,\xi=-\sqrt{R(\zeta)}\big);
\end{equation}
and in particular, for the infinity $\infty$ on the Riemann surface,
we denote the two corresponding points on $\mathcal R$ by
$\infty_+,\,\infty_-$.

Generically \cite{Farkas,Griffiths,Mumford}
based on the curve \eqref{R-curve}, one can introduce Abelian differentials of the first kind,
by which a Riemann theta function can be defined.
One can refer to Appendix \ref{B-1} for details.

Next, let us introduce our integrable map.
Using the Darboux transformation (\ref{eq:1.4}) (with \eqref{eq:2.2}), we define the following linear map:
\begin{subequations}\label{eq:3.9}
\begin{align}
&S_\beta:\mathbb{R}^{2N}\rightarrow\mathbb{R}^{2N},\quad (p,q)\mapsto(\t{p},\t{q}),\label{eq:3.9a}\\
&{\t{p}_j\choose\t{q}_j}=\frac{1}{\sqrt{\alpha_j^2-\beta^2}}
D^{(\beta)}(\alpha_j;a,b){p_j\choose q_j},\quad 1\leq j\leq N.\label{eq:3.9b}
\end{align}\end{subequations}
The extra factor in front of $D^{(\beta)}$ is to make the determinant to be 1,
which is useful in proving $S_{\beta}$ to be symplectic.
One can convert the constraint (\ref{eq:3.1c}) on $(u,v)$ to be the following constraint on $(a,b)$:
\begin{subequations}\label{eq:3.10}
\begin{align}
&P_1(a)\equiv(<Ap,p>b+1)a+<Ap,p>=0,\label{eq:3.10b}\\
&P(\beta b)\equiv(\beta b)^2L^{12}(\beta)-2(\beta b)L^{11}(\beta)-L^{21}(\beta)=0.\label{eq:3.10a}
\end{align}
\end{subequations}
In fact, \eqref{eq:3.10b} is nothing but
\[P_1(a)=Z(<Ap,p>+u),\]
where $a=Zu=(ab+1)u$ in \eqref{eq:2.2b} and $u=-<Ap,p>$ have been used.
To achieve \eqref{eq:3.10a}, one may start from the constraint $v=<Aq,q>$ and consider
\[P(\beta b)=\beta(<A\t{q},\t{q}>-\t{v}).\]
Replacing $\t v$ by $b=Z\t v$ from \eqref{eq:2.2b} and replacing $(\t p, \t q)$ using the map \eqref{eq:3.9b},
yield \eqref{eq:3.10a}.

In the following, by the nonlinearized map $S_{\beta}$ we mean the map \eqref{eq:3.9} after
imposing constraint \eqref{eq:3.10} on $(a,b)$. In fact, \eqref{eq:3.10} indicates that
$(a,b)$ can be explicitly expressed in terms of $(p,q)$,
which is denoted as
\begin{subequations}\label{eq:3.15}
\begin{equation}\label{eq:3.15a}
(a,b)=f_{\beta}(p,q)
\end{equation}
where
\begin{equation}\label{eq:3.15b}
a=\frac{-<Ap,p>}{1+<Ap,p>b},~~
b=\frac{-1}{\beta^2 Q_\beta(Ap,p)}\Big(\frac{1}{2}+Q_\beta(A^2p,q)\pm \mathcal{H}(\beta)\Big),
\end{equation}
\end{subequations}
with $\mathcal{H}(\beta)=-2 H(\beta)=\sqrt{-F(\beta)}=\displaystyle\frac{\sqrt{R(\beta^2)}}{2\alpha(\beta^2)}$
where $R(\beta^2)$ is defined by \eqref{eq:3.7}.
Note that $b$ is single-valued in light of the monodromy formulation \eqref{eq:3.8}. So is $a$.
Thus, after replacing $(a,b)$ in \eqref{eq:3.9} using \eqref{eq:3.15b}, the map is nonlinear with respect to $(p,q)$.

\begin{proposition}\label{P-3-1}
The nonlinearized map $S_{\beta}$ is symplectic and Liouville integrable, sharing the same
integrals  $E_1,\cdots,E_N$, defined by equation (\ref{eq:3.6}),
with the Hamiltonian system $(\mathbb{R}^{2N},{\rm d}p\wedge{\rm d}q, H_1)$.
\end{proposition}

\begin{proof}
Direct calculation from equation (\ref{eq:3.9b}) yields
\begin{equation}\label{eq:3.12}
\sum_{j=1}^N ({\rm d}\t{p}_j\wedge{\rm d}\t{q}_j-{\rm d}p_j\wedge{\rm d}q_j)
 =\frac{1}{2}{\rm d}\frac{P_1(a)}{ab+1}\wedge{\rm d}b+\frac{1}{2}{\rm d}a\wedge{\rm d}
 \Big(\frac{b^2P_1(a)}{ab+1}+\beta^{-1}P(\beta b)\Big),
\end{equation}
which vanishes when constraint \eqref{eq:3.10} makes sense.
This means the map $S_{\beta}$ is symplectic.
In addition, consider the Lax matrix $L$ given in equation (\ref{eq:3.2}) and the
Darboux matrix  $D^{(\beta)}(\lambda;a,b)$ given in \eqref{eq:1.4}.
In light of the constraint (\ref{eq:3.10}), it turns out that
\begin{equation}\label{eq:3.13}
\begin{split}
L(\lambda;&\t{p},\t{q})D^{(\beta)}(\lambda;a,b)-D^{(\beta)}(\lambda;a,b)L(\lambda;p,q)\\
&=\left(\begin{array}{cc}\beta^2b&\lambda\\
     \lambda b^2&b
     \end{array}\right)P_1(a)+
\left(\begin{array}{cc}
\beta^2a&0\\
\lambda(ab+1)&a
\end{array}\right)
\beta^{-1}P(\beta b)=0.
\end{split}\end{equation}
Thus $\det L(\lambda;\t{p},\t{q})=\det L(\lambda;p,q)$, which indicates
\begin{equation}\label{eq:3.14}
F(\lambda;\t{p},\t{q})=F(\lambda;p,q).
\end{equation}
It then follows from (\ref{eq:3.5})  that
$E_k(\t{p},\t{q})=E_k(p,q)$, which are invariants of the map $S_{\beta}$.

\end{proof}

\section{Algebro-geomitric solutions to the lpmKP equation}\label{sec-4}

In this section we proceed to derive algebro-geometric solutions
to the lpmKP equation \eqref{eq:1.1}.

First, using the integrable symplectic map \eqref{eq:3.9}, we define  discrete phase flow
$\big(p(m),\,q(m)\big)=S_{\beta}^m(p(0),q(0))$ with initial point $(p(0),q(0))\in \mathbb{R}^{2N}$,
and then we use \eqref{eq:3.15} to define finite genus potential $(a,b)$ in (\ref{eq:1.4}), i.e.
\begin{equation}\label{eq:4.1}
(a_m,b_m)=f_\beta\big(p(m),\,q(m)\big),
\end{equation}
which coincides with $u_m=-<Ap(m),p(m)>,\,v_m=<Aq(m),q(m)>$.
In light of \eqref{eq:2.2}, we also have
\begin{subequations}\label{eq:4.2}
\begin{align}
&a_m=Z_m u_m,\quad b_m=Z_m v_{m+1},\label{eq:4.2a}\\
&Z_m=a_m b_m+1=\frac{2}{1+\sqrt{1-4u_m v_{m+1}}}.\label{eq:4.2b}
\end{align}
\end{subequations}
We will finally reconstruct $Z_m$ in terms of theta function (see equation \eqref{eq:4.24})
and by ``integration'' from \eqref{eq:2.8} we recover the lpmKP solution $W$.

Next, consider the discrete KN spectral problem \eqref{eq:1.4} and the discrete Lax equation (\ref{eq:3.13})
along the $S_\beta^m$-flow, which are rewritten as
\begin{equation}\label{eq:4.5}
h(m+1,\lambda)=D_m(\lambda)h(m,\lambda)
\end{equation}
and
\begin{equation}\label{eq:4.3}
L_{m+1}(\lambda)D_m(\lambda)=D_m(\lambda)L_m(\lambda),
\end{equation}
where the Darboux matrix $D_m(\lambda)$ is
\begin{equation}\label{eq:4.4}
D_m(\lambda)=D^{(\beta)}(\lambda;a_m,b_m)
=\left(\begin{array}{cc}
\lambda^2 Z_m-\beta^2&\lambda Z_m u_m\\
\lambda Z_m v_{m+1}&1
\end{array}\right)
\end{equation}
and $L_m(\lambda)=L\big(\lambda;p(m),q(m)\big)$ is defined as the form \eqref{eq:3.2}.
Let $M(m,\lambda)=\Bigl(\begin{smallmatrix} M^{11}& M^{12}\\M^{21}& M^{22} \end{smallmatrix}\Bigr)$
be a fundamental solution matrix of \eqref{eq:4.5} with $M(0,\lambda)$ being the unit matrix $I$.
It turns out that
\begin{subequations}\label{eq:4.6}
\begin{align}
&M(m,\lambda)=D_{m-1}(\lambda)D_{m-2}(\lambda)\cdots D_0(\lambda),\label{eq:4.6a}\\
&L_m(\lambda)M(m,\lambda)=M(m,\lambda)L_0(\lambda),\label{eq:4.6b}
\end{align}
\end{subequations}
and we then have $\det M(m,\lambda)=(\zeta-\beta^2)^m$ due to $\det D_m(\lambda)=\zeta-\beta^2$.
For such an $M(m,\lambda)$ one can obtain its asymptotic behaviors from \eqref{eq:4.6a}.

\begin{lemma}\label{L-4.1}
$M^{11}(m,\lambda), \lambda
M^{12}(m,\lambda), \lambda M^{21}(m,\lambda)$ and $M^{22}(m,\lambda)$
are polynomials of $\zeta=\lambda^2$.
When $\zeta\sim\infty$, for $m\geq 2$ we have
\begin{subequations}\label{eq:4.7}
\begin{align}\label{eq:4.7a}
&M^{11}(m,\lambda)=Z_0Z_1\cdots Z_{m-1}\zeta^m+O(\zeta^{m-1}),\\\label{eq:4.7b}
&\lambda M^{12}(m,\lambda)=u_0Z_0Z_1\cdots Z_{m-1}\zeta^m+O(\zeta^{m-1}),\\\label{eq:4.7c}
&\lambda M^{21}(m,\lambda)=v_mZ_0Z_1\cdots Z_{m-1}\zeta^m+O(\zeta^{m-1}),\\\label{eq:4.7d}
&M^{22}(m,\lambda)=u_0v_mZ_0Z_1\cdots
Z_{m-1}\zeta^{m-1}+O(\zeta^{m-2}),
\end{align}\end{subequations}
and for $m=1$  they are still valid except $M^{22}(1,\lambda)=1$. When $\zeta\sim 0$, we have $(m\geq 1)$
\begin{equation}\label{eq:4.8}
\begin{split}
&M^{11}=(-\beta^2)^m+O(\zeta),\quad \lambda M^{12}=O(\zeta),\\
&\lambda M^{21}=O(\zeta),\quad M^{22}=1+O(\zeta).
\end{split}
\end{equation}
\end{lemma}

Equation (\ref{eq:4.3}) indicates that the solution space of equation (\ref{eq:4.5}) is invariant
under the action of the linear map $L_m(\lambda)$.
From \eqref{eq:3.2}, the traceless $L_m(\lambda)$ allows two opposite eigenvalues,
denoted as $\pm \mathcal{H}(\lambda)=\pm \sqrt{-F(\lambda)}$,
which are independent of the discrete argument $m$ due to relation (\ref{eq:3.14}).
Denoting the corresponding eigenvectors by $h_{\pm}(m,\lambda)=(h_{\pm}^{(1)}, h_{\pm}^{(2)})^T$,
we have
\begin{subequations}\label{eq:4.9}
\begin{equation}\label{eq:4.9a}
L_m(\lambda)h_{\pm}(m,\lambda)=\pm \mathcal{H}(\lambda)h_{\pm}(m,\lambda),
\end{equation}
and
\begin{equation}\label{eq:4.9b}
h_{\pm}(m+1,\lambda)=D_m(\lambda)h_{\pm}(m,\lambda),
\end{equation}
\end{subequations}
simultaneously.
Noting that the rank of $L_m(\lambda)\mp \mathcal{H}(\lambda)$ is 1,
which means in each case the common eigenvector is uniquely determined up to a constant factor,
we select two eigenvectors $h_{\pm}(m,\lambda)$ defined through  $M(m,\lambda)$,  as the following,
\begin{equation}\label{eq:4.10}
h_{\pm}(m,\lambda)={h_{\pm}^{(1)}\choose
h_{\pm}^{(2)}}=M(m,\lambda){c_\lambda^\pm\choose 1},
\end{equation}
where the constants $c_\lambda^\pm$ are determined through \eqref{eq:4.9a} and \eqref{eq:4.6b}, by
\[L_0(\lambda) \begin{pmatrix}c_\lambda^{+}& c_\lambda^{-}\\1&1\end{pmatrix}
=\begin{pmatrix}c_\lambda^{+}& c_\lambda^{-}\\1&1\end{pmatrix}
\begin{pmatrix}\mathcal{H}(\lambda)& 0\\ 0 & -\mathcal{H}(\lambda)\end{pmatrix},\]
i.e. taking $m=0$ in equation (\ref{eq:4.9a}).
It turns out that
\begin{equation}\label{eq:4.11}
c_\lambda^\pm=\frac{L_0^{11}(\lambda)\pm
\mathcal{H}(\lambda)}{L_0^{21}(\lambda)}
=\frac{-L_0^{12}(\lambda)}{L_0^{11}(\lambda)\mp \mathcal{H}(\lambda)}.
\end{equation}

Next, we will investigate these eigenvectors $h_{\pm}(m,\lambda)$ using the Baker--Akhiezer functions,
which can be expressed by theta function on the hyperelliptic Riemann surface corresponding to the
spectral curve $\mathcal{R}$.
Let us introduce the Baker--Akhiezer functions, which are meromorphic on $\mathcal R$, by
\begin{equation}\label{eq:4.12}
\begin{split}
&\mathfrak{h}^{(1)}\big(m,\mathfrak p_+(\lambda^2)\big)=\lambda h_+^{(1)}(m,\lambda),\quad
\mathfrak{h}^{(1)}\big(m, \mathfrak p_-(\lambda^2)\big)=\lambda h_-^{(1)}(m,\lambda),\\
&\mathfrak{h}^{(2)}\big(m,\mathfrak p_+(\lambda^2)\big)=h_+^{(2)}(m,\lambda),\quad
\mathfrak{h}^{(2)}\big(m, \mathfrak p_-(\lambda^2)\big)=h_-^{(2)}(m,\lambda).
\end{split}\end{equation}
To associate them with the Riemann theta function (see  (\ref{eq:B9})),
we investigate their analytic behaviors and divisors.
To this end,  introduce elliptic variables $\mu_j,\,\nu_j$ in $L^{12}$ and $L^{21}$ by
\begin{subequations}\label{eq:4.13}
\begin{align}\label{eq:4.13a}
&\lambda^{-1}L_m^{12}(\lambda)=-Q_\lambda(Ap(m),p(m))=\frac{u_m}{\alpha(\zeta)}\prod_{j=1}^{N-1}
\big(\zeta-\mu_j^2(m)\big),\\\label{eq:4.13b}
&\lambda^{-1}L_m^{21}(\lambda)=Q_\lambda(Aq(m),q(m))=\frac{v_m}{\alpha(\zeta)}\prod_{j=1}^{N-1}
\big(\zeta-\nu_j^2(m)\big),
\end{align}
\end{subequations}
which lead us,  from equations (\ref{eq:4.6a}) and (\ref{eq:4.6b}), to
\begin{subequations}\label{eq:4.14}
\begin{align}
&\mathfrak{h}^{(1)}\big(m,\mathfrak p_+(\lambda^2)\big)\cdot
\mathfrak{h}^{(1)}\big(m, \mathfrak p_-(\lambda^2)\big)
=\zeta(\zeta-\beta^2)^m\frac{-u_m}{v_0}\prod_{j=1}^{N-1}\frac{\zeta-\mu_j^2(m)}{\zeta-\nu_j^2(0)},
\label{eq:4.14a}\\
&\mathfrak{h}^{(2)}\big(m,\mathfrak p_+(\lambda^2)\big)\cdot
\mathfrak{h}^{(2)}\big(m, \mathfrak p_-(\lambda^2)\big)
=(\zeta-\beta^2)^m\frac{v_m}{v_0}\prod_{j=1}^{N-1}\frac{\zeta-\nu_j^2(m)}{\zeta-\nu_j^2(0)}.\label{eq:4.14b}
\end{align}
\end{subequations}
With regard to asymptotic behaviors of the Baker--Akhiezer functions, we have the following lemmas.

\begin{lemma}\label{L-4.2}
The Baker--Akhiezer functions \eqref{eq:4.12} have the following asymptotic behaviors
as $\zeta=\lambda^2\sim\infty$,
\begin{subequations}\label{eq:4.15}
\begin{align}
&\mathfrak{h}^{(1)}\big(m,\mathfrak p_+(\lambda^2)\big)
=\frac{1}{2v_0}Z_0Z_1\cdots Z_{m-1}\zeta^{m+1}\big(1+O(\zeta^{-1})\big),\label{eq:4.15a}\\
&\mathfrak{h}^{(1)}\big(m, \mathfrak p_-(\lambda^2)\big)
=\frac{-2u_m}{Z_0Z_1\cdots Z_{m-1}}\big(1+O(\zeta^{-1})\big),\label{eq:4.15b}\\
&\mathfrak{h}^{(2)}\big(m,\mathfrak p_+(\lambda^2)\big)
=\frac{v_m}{2v_0}Z_0Z_1\cdots Z_{m-1}\zeta^m\big(1+O(\zeta^{-1})\big),\label{eq:4.15c}\\
&\mathfrak{h}^{(2)}\big(m, \mathfrak p_-(\lambda^2)\big)
=\frac{2}{Z_0Z_1\cdots Z_{m-1}}\big(1+O(\zeta^{-1})\big).\label{eq:4.15d}
\end{align}\end{subequations}
\end{lemma}

\begin{proof}
First, as $\lambda\sim\infty$ we find
$\lambda c_\lambda^+=(\zeta/2v_0)\big(1+O(\zeta^{-1})\big)$ from equations (\ref{eq:4.11}) and (\ref{eq:3.2}).
Then, using (\ref{eq:4.10}) and the asymptotic results of $M(m,\lambda)$ given in Lemma \ref{L-4.1},
one can obtain  (\ref{eq:4.15a}) and (\ref{eq:4.15c}).
The other two in \eqref{eq:4.15} follow from  \eqref{eq:4.14}.

\end{proof}

\begin{lemma}\label{L-4.3}
When $\zeta=\lambda^2\sim 0$, the following asymptotic behaviors hold,
\begin{subequations}\label{eq:4.16}
\begin{align}\label{eq:4.16a}
&\mathfrak{h}^{(1)}\big(m,\mathfrak p_+(\lambda^2)\big)
=(-\beta^2)^m\frac{1-2<p(0),q(0)>}{-<A^{-1}q(0),q(0)>}\big(1+O(\zeta)\big),\\
&\mathfrak{h}^{(1)}\big(m, \mathfrak p_-(\lambda^2)\big)
=\zeta\frac{<A^{-1}q(0),q(0)>
u_m}{(1-2<p(0),q(0)>)v_0}\left(\prod_{j=1}^{N-1}\frac{\mu_j^2(m)}{\nu_j^2(0)}\right)\big(1+O(\zeta)\big).
\label{eq:4.16b}
\end{align}
\end{subequations}
\end{lemma}

\begin{proof} From (\ref{eq:4.11}) we have $\lambda
c_\lambda^+=-\frac{1-2<p(0),q(0)>}{<A^{-1}q(0),q(0)>}\big(1+O(\zeta)\big)$ as $\lambda\sim 0$.
Equations (\ref{eq:4.8}) and  (\ref{eq:4.10}) yield (\ref{eq:4.16a}),
which further gives rise to \eqref{eq:4.16b} by using  (\ref{eq:4.14a}).

\end{proof}

Now we are able to write down divisors of the Baker--Akhiezer functions
$\mathfrak h^{(1)}(m,\mathfrak p),\,\mathfrak h^{(2)}(m,\mathfrak p)$ on $\mathcal{R}$,
which are, respectively, (cf.\cite{Farkas,Griffiths,Mumford})
\begin{subequations}\label{eq:4.17}
\begin{align}
&\mathcal{D}(\mathfrak h^{(1)}(m,\mathfrak p))=\sum_{j=1}^g\Big(\mathfrak p\big(\mu_j^2(m)\big)-\mathfrak
p\big(\nu_j^2(0)\big)\Big)+\{\mathfrak o_-\}+m\{\mathfrak p(\beta^2)\}-(m+1)\{\infty_+\},\label{eq:4.17a}\\
&\mathcal{D}(\mathfrak h^{(2)}(m,\mathfrak p))=\sum_{j=1}^g\Big(\mathfrak p\big(\nu_j^2(m)\big)-\mathfrak
p\big(\nu_j^2(0)\big)\Big)+m\{\mathfrak p(\beta^2)\}-m\{\infty_+\},\label{eq:4.17b}
\end{align}\end{subequations}
where $\mathfrak o_-=(\zeta=0,\,\xi=-\sqrt{R(0)})$, $ g=N-1$.

Next, introduce the Abel--Jacobi variables
\begin{equation}\label{eq:4.18}
\vec{\psi}(m)=\mathcal{A}\Big(\hbox{$\sum$}_{j=1}^g\mathfrak p\big(\mu_j^2(m)\big)\Big), \quad
\vec{\phi}(m)=\mathcal{A}\Big(\hbox{$\sum$}_{j=1}^g\mathfrak p\big(\nu_j^2(m)\big)\Big),
\end{equation}
by using the Abel map $\mathcal{A}$  (see Appendix \ref{B-1}).
Employing Toda's dipole technique \cite{Toda},
from \eqref{eq:4.18} and \eqref{eq:4.17} we have
\begin{subequations}\label{eq:4.19}
\begin{align}
\label{eq:4.19a}
&\vec\psi(m)\equiv\vec\phi(0)+m\vec\Omega_\beta+\vec\Omega_0,\quad({\rm mod}\,\mathcal T),\\
\label{eq:4.19b}
&\vec\phi(m)\equiv\vec\phi(0)+m\vec\Omega_\beta,\quad({\rm mod}\,\mathcal T),\\
\label{eq:4.19c}
&\vec\Omega_\beta=\int_{\mathfrak p(\beta^2)}^{\infty_+}\vec\omega,\quad
\vec\Omega_0=\int_{\mathfrak o_-}^{\infty_+}\vec\omega.
\end{align}\end{subequations}
Then, as usual treatment (cf.\cite{CaoX-JPA-2012,Farkas,Griffiths,Mumford}),
by comparing divisors
we obtain express the Baker-Akhiezer functions   in terms of the Riemann theta function \eqref{eq:B9}:
\begin{subequations}\label{eq:4.20}
\begin{align}\label{eq:4.20a}
&\mathfrak h^{(1)}(m,\mathfrak p)=C_m^{(1)}
\frac{\theta(-\mathcal A(\mathfrak p)+\vec\psi(m)+\vec K;B)}
{\theta(-\mathcal A(\mathfrak p)+\vec\phi(0)+\vec K;B)}
\exp\int_{\mathfrak p_0}^{\mathfrak p}(m\,\omega[\mathfrak p(\beta^2),\infty_+]
+\omega[\mathfrak o_-,\infty_+]),\\
\label{eq:4.20b}
&\mathfrak h^{(2)}(m,\mathfrak p)=C_m^{(2)}
\frac{\theta(-\mathcal A(\mathfrak p) +\vec\phi(m)+\vec K;B)}
{\theta(-\mathcal A(\mathfrak p)+\vec\phi(0)+\vec K;B)}
\exp\int_{\mathfrak p_0}^{\mathfrak p}m\,\omega[\mathfrak p(\beta^2),\infty_+],
\end{align}
\end{subequations}
where $C_m^{(1)}$ and $C_m^{(2)}$ are constant factors
and  the Riemann constant vector $\vec K$ is defined in \eqref{eq:B10}.
Here, $\omega[\mathfrak p,\mathfrak q]$ is the dipole, a meromorphical differential that has only simple
poles at $\mathfrak p$ and $\mathfrak q$  with residues $+1$ and $-1$,  respectively.
%

Our purpose is to derive explicit expression of $Z_m$ in terms of the Riemann theta function.
To achieve that, first, we take $\mathfrak p\rightarrow\infty_-$ in equation (\ref{eq:4.20b})
and then compare the result with the asymptotic formula (\ref{eq:4.15d}).
This gives rise to
\begin{equation}\label{eq:4.21}
C_m^{(2)}=\frac{2}{Z_0Z_1\cdots Z_{m-1}}
\frac{\theta[\vec\phi(0)+\vec K+\vec\eta_{\infty_-}]}{\theta[\vec\phi(m)+\vec K+\vec\eta_{\infty_-}]}
\exp\int_{\infty_-}^{\mathfrak p_0}m\,\omega[\mathfrak p(\beta^2),\infty_+],
\end{equation}
where $\vec\eta_{\infty_-}=-\mathcal A(\infty_-)$.
Next, we consider the second row in equation (\ref{eq:4.9b}), i.e.
\begin{equation}\label{eq:4.22}
\mathfrak h^{(2)}(m+1,\mathfrak p)=b_m\mathfrak h^{(1)}(m,\mathfrak p)+\mathfrak h^{(2)}(m,\mathfrak p),
\end{equation}
which reads
\begin{equation}\label{eq:4.22-0}
\mathfrak h^{(2)}(m+1,\mathfrak o_-)=\mathfrak h^{(2)}(m,\mathfrak o_-)
\end{equation}
at the point $\mathfrak o_-$ since $\mathfrak h^{(1)}(m,\mathfrak o_-)=0$.
Substituting \eqref{eq:4.20b} with $\mathfrak p=\mathfrak o_-$ into \eqref{eq:4.22-0} immediately yields
\begin{equation}\label{eq:4.23}
\frac{C_m^{(2)}}{C_{m+1}^{(2)}}=
\frac{\theta(\vec\phi(m+1)+\vec K +\vec\eta_{\mathfrak o_-};B)}
{\theta(\vec\phi(m)+\vec K+\vec\eta_{\mathfrak o_-};B)}
\exp\int_{\mathfrak p_0}^{\mathfrak o_-}\omega[\mathfrak p(\beta^2),\infty_+],
\end{equation}
where $\vec\eta_{\mathfrak o_-}=-\mathcal A(\mathfrak o_-)$.
Now, substituting \eqref{eq:4.21} into the above equation, we arrive at an explicit expression of $Z_m$
in terms of theta function, i.e.
\begin{equation}\label{eq:4.24}
Z_m=\frac{\theta(\vec\phi(m+1)+\vec K+\vec\eta_{\mathfrak o_-};B)
\cdot\theta(\vec\phi(m) +\vec K+\vec\eta_{\infty_-};B)}
{\theta(\vec\phi(m+1)+\vec K+\vec\eta_{\infty_-};B)
\cdot \theta(\vec\phi(m)+\vec K+\vec\eta_{\mathfrak o_-};B)}
\exp\int_{\infty_-}^{\mathfrak o_-}\omega[\mathfrak p(\beta^2),\infty_+].
\end{equation}
With $Z_m$ in hand, for a function $W_m$ that obeys equation $W_{m+1}-W_m=\ln Z_m$
where $Z_m$ is given in \eqref{eq:4.24}, one can obtain an explicit solution by ``integration'',
\begin{equation}\label{eq:4.25}
W_m=W_0+\ln\frac{\theta[\vec\phi(m)+\vec K+\vec\eta_{\mathfrak o_-}]\cdot\theta[\vec\phi(0)+
\vec K+\vec\eta_{\infty_-}]}{\theta[\vec\phi(m)+\vec K+\vec\eta_{\infty_-}]\cdot\theta[\vec\phi(0)
+\vec K+\vec\eta_{\mathfrak o_-}]}+m\int_{\infty_-}^{\mathfrak o_-}\omega[\mathfrak p(\beta^2),\infty_+].
\end{equation}

\begin{remark}\label{R-4.1}
Explicit expression for $u_m$ and $v_m$ can be computed with the help of equations
(\ref{eq:4.15b}), (\ref{eq:4.15c}), (\ref{eq:4.20}) and (\ref{eq:4.24}).
This will lead to the finite-gap solutions of the ldNLS equation \eqref{eq:2.9}, in a similar way to the following.
\end{remark}

The above discussions and results are valid for
$(m,\beta)=(m_i,\beta_i)$, $i=1,2,3$.
Thus, we have three integrable symplectic maps $S_{\beta_1},\,S_{\beta_2}$ and $S_{\beta_3}$,
which  commute with each other since they share the same Liouville integrals $E_1,\cdots,E_N$
(cf.\cite{Bruschi,CZ2012a,CZ2012b,Veselov,Veselov1}).
This enables us to define
\begin{equation}\label{eq:4.26}
\big(p(m_1,m_2,m_3),\,q(m_1,m_2,m_3)\big)=S_{\beta_1}^{m_1}S_{\beta_2}^{m_2}S_{\beta_3}^{m_3}
(p(0,0,0),q(0,0,0)),
\end{equation}
and the components $(p_j,q_j)$ satisfy the equation (\ref{eq:3.9b}) for $\beta=\beta_1,\beta_2,\beta_3$ simultaneously
in the case of $\lambda=\alpha_j$.
This leads to a compatible solution
\begin{equation}\label{eq:4.27}
\chi(m_1,m_2,m_3)=(\alpha_j^2-\beta_1^2)^{\frac{m_1}{2}}
(\alpha_j^2-\beta_2^2)^{\frac{m_2}{2}}(\alpha_j^2-\beta_3^2)^{{\frac{m_3}{2}}}(p_j,q_j)
\end{equation}
that satisfies the three equations in (\ref{eq:2.6}) with spectral parameter $\lambda=\alpha_j$.
Thus, one can obtain $W(m_1,m_2,m_3)$ by ``integrating'' \eqref{eq:2.6},
and in light of Theorem \ref{T-2-1}, such $W(m_1,m_2,m_3)$ provides solutions to the lpmKP equation \eqref{eq:1.1}.

\begin{theorem}\label{T-4}
The lpmKP equation (\ref{eq:1.1}) admits algebro-geometric solutions expressed in terms of the Riemann theta function,
\begin{align}\label{eq:4.28}
W(m_1,m_2,m_3)=&\ln\frac{\theta(\sum_{k=1}^3m_k\vec\Omega_{\beta_k}+\vec\phi(0,0,0)
+\vec K+\vec\eta_{\mathfrak o_-};B)\cdot
\theta(\vec\phi(0,0,0) +\vec K+\vec\eta_{\infty_-};B)}
{\theta(\sum_{k=1}^3m_k\vec\Omega_{\beta_k}+\vec\phi(0,0,0) +\vec K+\vec\eta_{\infty_-};B)\cdot \theta(\vec\phi(0,0,0)+\vec K+\vec\eta_{\mathfrak o_-};B)}\nonumber\\
&+\sum_{k=1}^3 m_k\int_{\infty_-}^{\mathfrak o_-}\omega[\mathfrak p(\beta_k^2),\infty_+]+W(0,0,0),
\end{align}
where  $\vec\eta_{\mathfrak o_-}=-\mathcal A(\mathfrak o_-)$, $\vec\eta_{\infty_-}=-\mathcal A(\infty_-)$,
$\vec K$ is the Riemann constant vector \eqref{eq:B10},
\begin{equation}
\vec\Omega_{\beta_k}=\int_{\mathfrak p(\beta^2_k)}^{\infty_+}\vec\omega,
\end{equation}
and the dipole differential $\omega[\mathfrak p(\beta_k^2),\infty_+]$ is defined as
 \begin{equation}\label{eq:E1}
\omega[\mathfrak p(\beta_k^2),\infty_+]=
\left(\zeta+ \frac{\xi+\sqrt{R(\beta_{k}^{2})}}{\zeta-\beta_{k}^{2}}\right)
\frac{\mathrm{d}\zeta}{2\sqrt{R(\zeta)}},
\end{equation}
with $\xi$ and  $R(\zeta)$ given by \eqref{R-curve}.
\end{theorem}

\section{An example: $g=1$ case}\label{sec-5}

As an example of the solution  \eqref{eq:4.28},
in the following we explore the simplest case, where genus $g=1$.
The elliptic curve $\mathcal R$   \eqref{R-curve} reads
\begin{equation}\label{eq:C1}
\xi^2=R(\zeta)=(\zeta-\zeta_{1})(\zeta-\zeta_{2})(\zeta-\zeta_{3})(\zeta-\zeta_{4}),
\end{equation}
with $\zeta_{1}=\lambda_1^2, \zeta_{2}=\lambda_2^2, \zeta_{3}=\alpha_1^2, \zeta_{4}=\alpha_2^2$.

In our example, we assume all $\{\zeta_j\}$ are on real axis
and $\zeta_{1}<\zeta_{2}<\zeta_{3}<\zeta_{4}$ so that  related branch cuts
are taken as $[\zeta_{1},\zeta_{2}]$ and $[\zeta_{3},\zeta_{4}]$.
Thus, we have the Abelian differential of the first kind
\begin{equation}\label{eq:C2}
\omega_1=\displaystyle \frac{C_{11}\mathrm{d}\zeta}{2\sqrt{(\zeta-\zeta_{1})(\zeta-\zeta_{2})(\zeta-\zeta_{3})(\zeta-\zeta_{4})}},
\end{equation}
where the normalization constant $C_{11}$ is
\begin{equation}\label{eq:C3}
C_{11}^{-1}=\int_{a_1}\displaystyle \frac{\mathrm{d}\zeta}{2\sqrt{(\zeta-\zeta_{1})(\zeta-\zeta_{2})(\zeta-\zeta_{3})(\zeta-\zeta_{4})}}
\end{equation}
along with an $a_1$-period.
Then we have the two periods
\begin{subequations}\label{eq:C4}
\begin{align}
1=&\int_{a_1}\omega_1=\int_{\zeta_3}^{\zeta_4}\displaystyle \frac{C_{11}\mathrm{d}\zeta}{\sqrt{(\zeta-\zeta_{1})(\zeta-\zeta_{2})(\zeta-\zeta_{3})(\zeta-\zeta_{4})}},
\label{eq:C4a}\\
B_{11}=&\int_{b_1}\omega_1=\int_{\zeta_2}^{\zeta_3}\displaystyle \frac{C_{11}\mathrm{d}\zeta}{\sqrt{(\zeta-\zeta_{1})(\zeta-\zeta_{2})(\zeta-\zeta_{3})(\zeta-\zeta_{4})}}.
\label{eq:C4b}
\end{align}
\end{subequations}
Note that by certain linear fractional transformation (see \cite{Toda})
the above two formulae can be converted to the elliptic integrals of the first kind:
\begin{subequations}\label{eq:C10}
\begin{align}
1=&\int_{a_1}\omega_1=A_{0}\int_{1}^{\frac{1}{\kappa}}\displaystyle \frac{\mathrm{d}s}{\sqrt{(1-s^{2})(1-\kappa^{2}s^{2})}},\label{eq:C10a}\\
B_{11}=&\int_{b_1}\omega_1=A_{0}\int_{-1}^{1}\displaystyle \frac{\mathrm{d}s}{\sqrt{(1-s^{2})(1-\kappa^{2}s^{2})}}, \label{eq:C10b}
\end{align}
\end{subequations}
where $\kappa\in (0,1)$ is a constant.

In this case the Riemann theta function reduces to the Jacobi theta function $\vartheta_{3}$,
\begin{equation}\label{eq:C11}
\theta(z; B_{11})=\sum_{n=-\infty}^{+\infty}\exp[\pi\sqrt{-1}(n^{2}B_{11}+2nz)]=\vartheta_{3}(z\mid B_{11}),\ \ z\in \mathbb{C}.
\end{equation}
Hence, the algebro-geometric solution \eqref{eq:4.28} in the case of $g=1$ can be expressed as
\begin{align}\label{eq:C12}
W(m_1,m_2,m_3)=&\ln\frac{\vartheta_{3}(\sum_{k=1}^3m_k\Omega_{\beta_k}+\phi(0,0,0)
+ K_1+\eta_{\mathfrak o_-}|B_{11})\cdot
\vartheta_{3}(\phi(0,0,0) + K_1+\eta_{\infty_-}|B_{11})}
{\vartheta_{3}(\sum_{k=1}^3m_k\Omega_{\beta_k}+\phi(0,0,0) + K_1+\eta_{\infty_-}|B_{11})\cdot
\vartheta_{3}(\phi(0,0,0)+ K_1+\eta_{\mathfrak o_-}|B_{11})}\nonumber\\
&+\sum_{k=1}^3 m_k\int_{\infty_-}^{\mathfrak o_-}\omega[\mathfrak p(\beta_k^2),\infty_+]+W(0,0,0),
\end{align}
where
\begin{subequations}\label{eq:C13}
\begin{align}
&K_1=-\int_{a_1}\mathcal A\,\omega_1+\displaystyle\frac{B_{11}}{2}, \ \ \mathcal A=\mathcal A(\mathfrak{p})=\int_{\mathfrak{p}_{0}}^{\mathfrak{p}}\omega_1,\\
&\Omega_{\beta_k}=\int_{\mathfrak p(\beta_{k}^2)}^{\infty_+}\omega_1,\ \
\eta_{\mathfrak o_-}=-\int_{\mathfrak{p}_{0}}^{\mathfrak o_-}\omega_1, \ \ \eta_{\infty_-}=-\int_{\mathfrak{p}_{0}}^{\infty_-}\omega_1,\label{eq:C13b}\\
&\omega[\mathfrak p(\beta_k^2),\infty_+]=\displaystyle\frac{1}{C_{11}}
\left(\zeta+\displaystyle\frac{\xi+\sqrt{R(\beta_{k}^{2})}}{\zeta-\beta_{k}^{2}}\right)\omega_1. \label{eq:C13c}
\end{align}
\end{subequations}
Note that due to the arbitrariness of $\phi(0,0,0)$ we can always vanish
$\phi(0,0,0)+K_1$ and thus we come to
\begin{subequations}\label{W}
\begin{align}\label{W-sol}
W(m_1,m_2,m_3)=W_2(m_1,m_2,m_3)+W_1(m_1,m_2,m_3)
\end{align}
with
\begin{align}\label{W2}
& W_2(m_1,m_2,m_3)= \ln\frac{\vartheta_{3}(\sum_{k=1}^3m_k\Omega_{\beta_k}+\eta_{\mathfrak o_-}|B_{11})
\cdot \vartheta_{3}(\eta_{\infty_-}|B_{11})}
{\vartheta_{3}(\sum_{k=1}^3m_k\Omega_{\beta_k}+\eta_{\infty_-}|B_{11})\cdot
\vartheta_{3}(\eta_{\mathfrak o_-}|B_{11})},\\
\label{W1}
& W_1(m_1,m_2,m_3)=\sum_{k=1}^3 m_k\int_{\infty_-}^{\mathfrak o_-}
\omega[\mathfrak p(\beta_k^2),\infty_+]+W(0,0,0),
\end{align}
\end{subequations}
where $\Omega_{\beta_k}, \eta_{\mathfrak o_-},  \eta_{\infty_-}$ and $\omega[\mathfrak p(\beta_k^2),\infty_+]$
are computed from \eqref{eq:C13b},
and $W_1(m_1,m_2,m_3)$  acts as a linear background of $W(m_1,m_2,m_3)$.

To illustrate the solution, we take
\begin{equation}\label{zeta-j}
\zeta_1=1,~\zeta_2= 3,~\zeta_3=5,~\zeta_4=8,~ \beta_1^2=10,~\beta_2^2=13,~\beta_3^2=14,
\end{equation}
and consequently,
\[\mathfrak{p}(\beta_1^2)=(10, 25.0998),~~ \mathfrak{p}(\beta_2^2)=(13, 69.282),~~
\mathfrak{p}(\beta_3^2)=(14, 87.8749).\]
It follows from \eqref{eq:C3}, \eqref{eq:C4} and \eqref{eq:C13} that
(the integrals are computed numerically using Mathematica)
\begin{equation}\label{parameters}
\begin{split}
&C_{11}=1.30467\,i, ~~   B_{11}=1.21091\,i, ~~ \eta_{\mathfrak o_-}=-0.391547\,i, ~~ \eta_{\infty_-}=-0.569795\,i,~~  \\
&\Omega_{\beta_1}=0.123456\,i, ~~ \Omega_{\beta_2}=0.0769913\,i, ~~\Omega_{\beta_3}=0.068633\,i,~~\\
&\int_{\infty_-}^{\mathfrak o_-}\omega[\mathfrak p(\beta_k^2),\infty_+]=\gamma_k,~~~
\gamma_1=4.80874,~~
\gamma_2=4.96645,~~
\gamma_3=5.03085,
\end{split}
\end{equation}
where we have taken $\mathfrak{p}_0=(-3.0, 45.9565)$ in \eqref{eq:C13}.
The quasi-periodic evolution of $W_2(m_1,m_2,m_3)$ is shown in Figure \ref{Fig-1}.
\begin{figure}[!h]
\centering
\begin{minipage}{7cm}
\includegraphics[width=\textwidth]{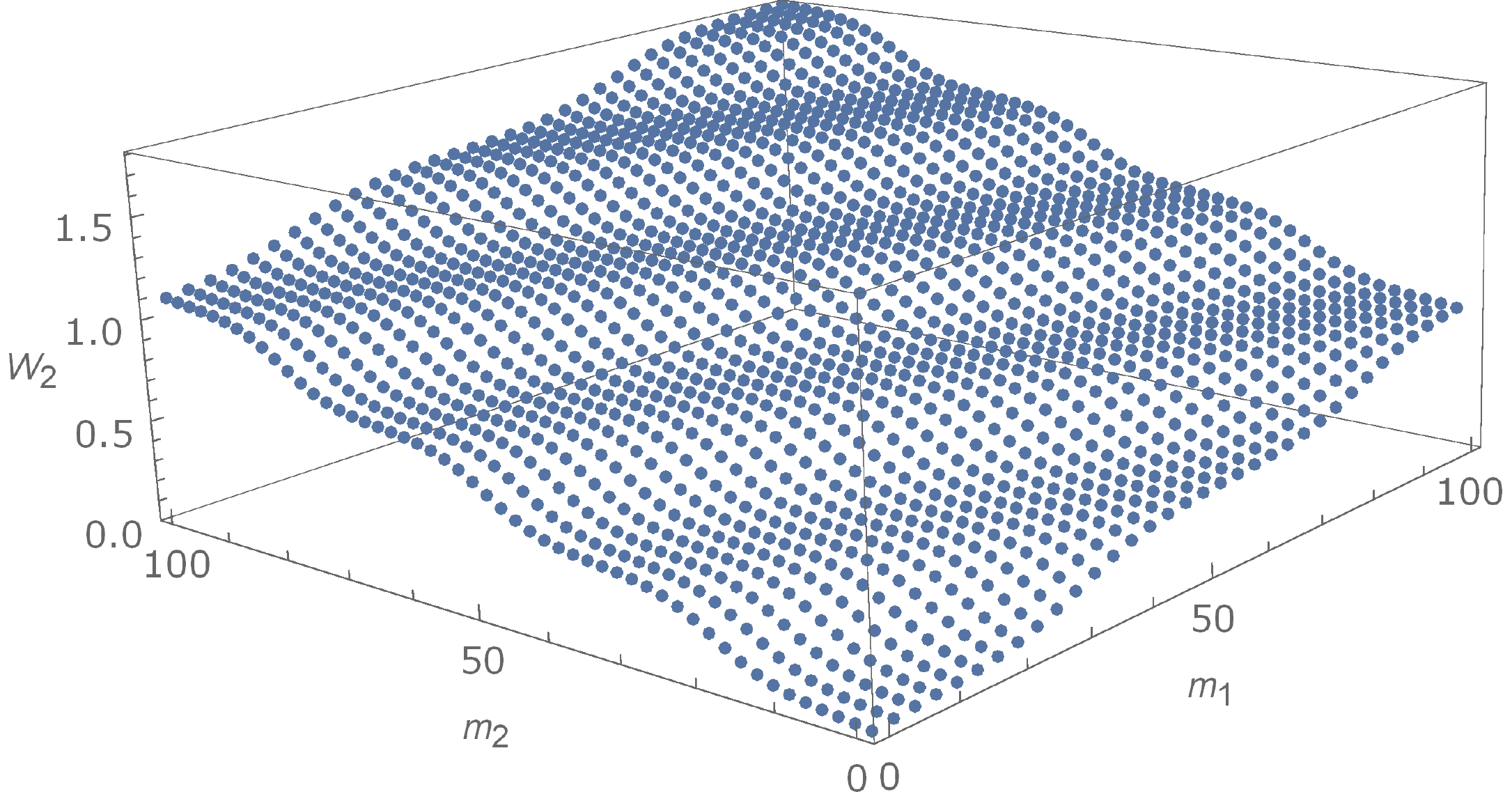}
\centering
\footnotesize{(a)}
\end{minipage}
\hspace{0.2in}
\begin{minipage}{6cm}
\includegraphics[width=\textwidth]{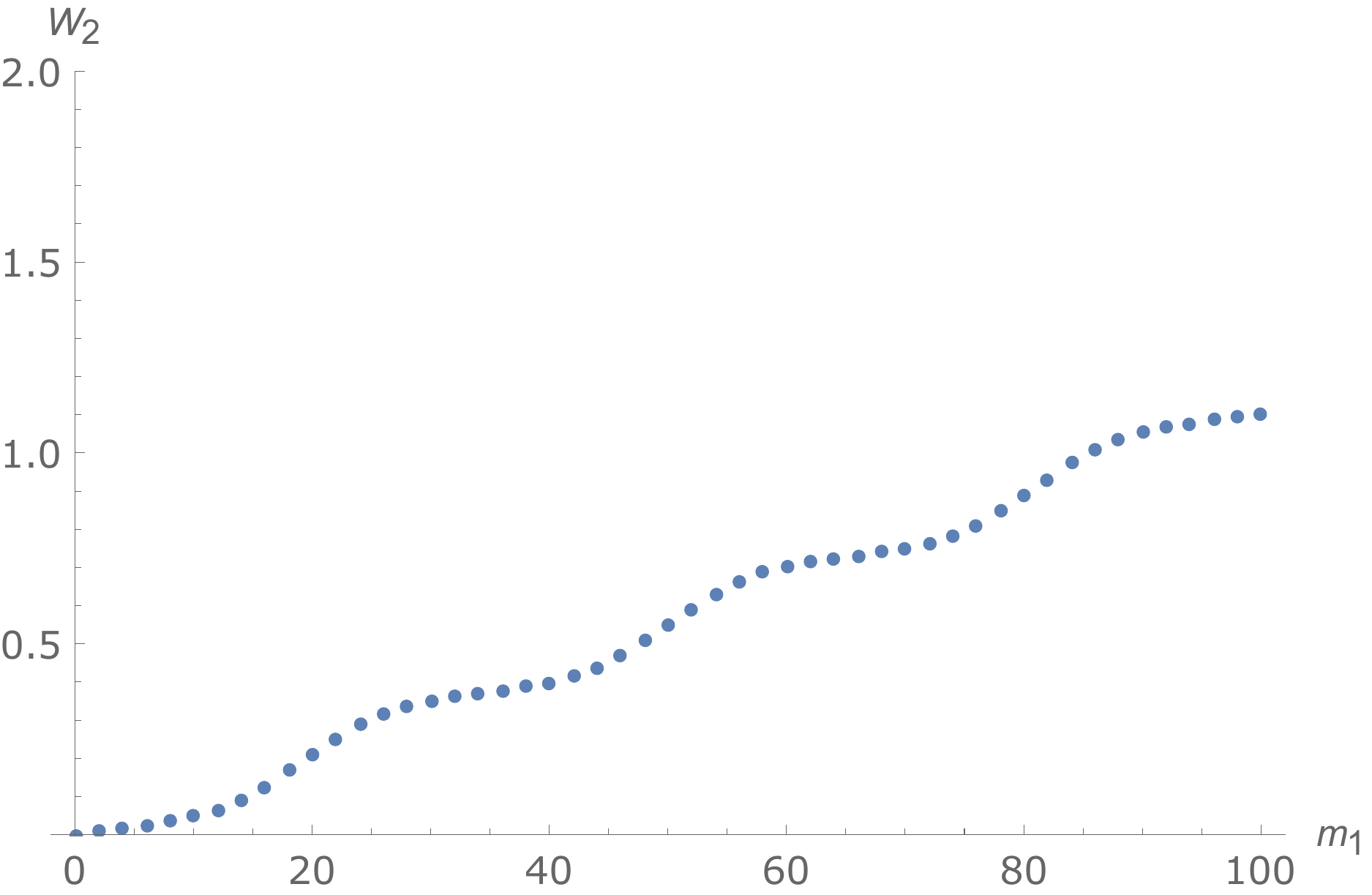}
\centering
\footnotesize{(b)}
\end{minipage}
\begin{minipage}{6cm}
\includegraphics[width=\textwidth]{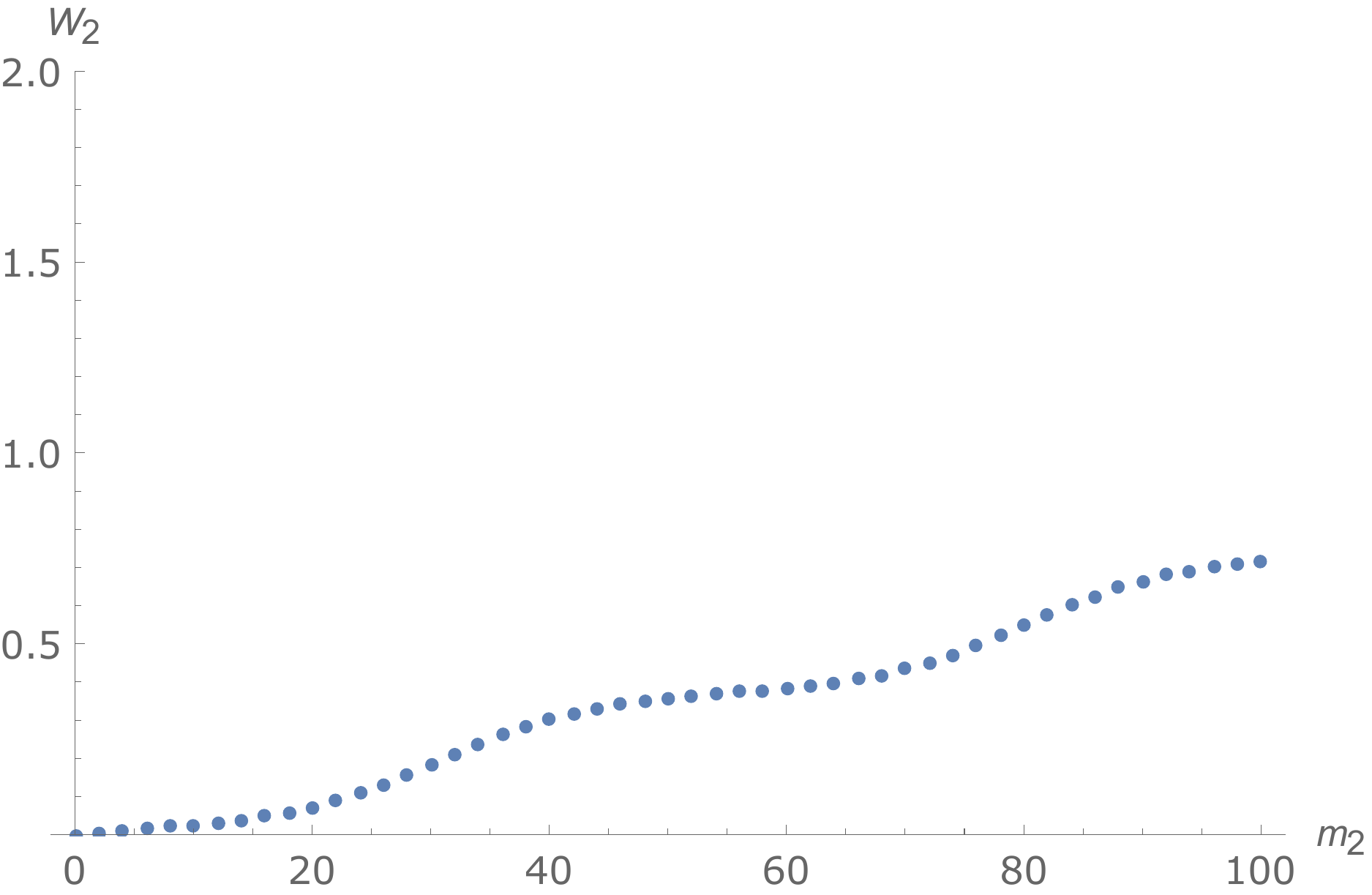}
\centering
\footnotesize{(c)}
\end{minipage}
\hspace{0.2in}
\begin{minipage}{6cm}
\includegraphics[width=\textwidth]{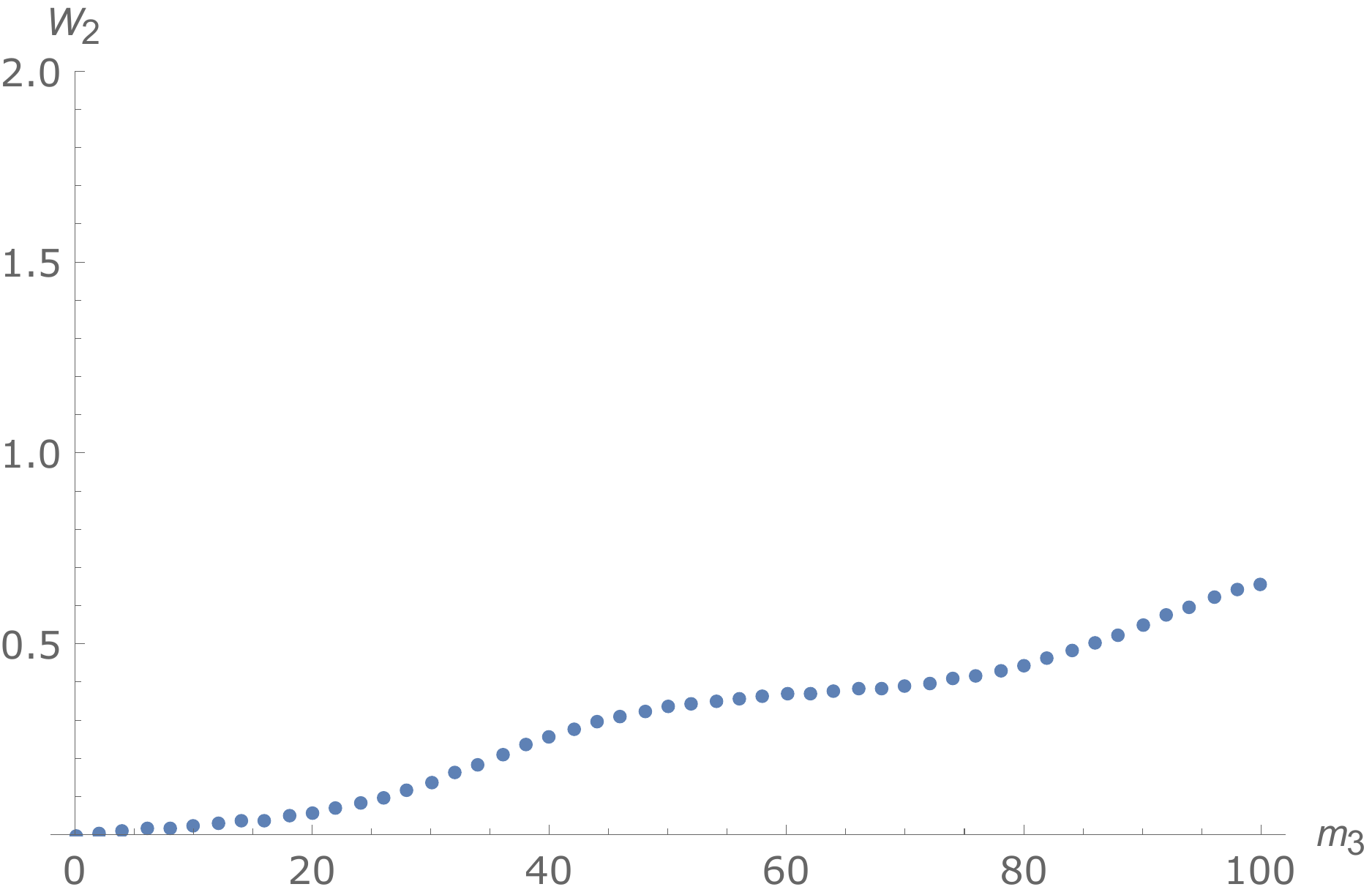}
\centering
\footnotesize{(d)}
\end{minipage}
\caption{Shape and motion of $W_2(m_1,m_2,m_3)$ given in \eqref{W2} for
$\{\zeta_j\}$ in \eqref{zeta-j} and $\mathfrak{p}_0=(-3.0, 45.9565)$.
(a) 3D plot of $W_2(m_1,m_2,0)$. (b) 2D plot of $W_2(m_1,0,0)$.
(c) 2D plot of $W_2(0,m_2,0)$.
(d) 2D plot of $W_2(0,0,m_3)$. }
\label{Fig-1}
\end{figure}

One can see a periodic wave coupled with an apparent linear background that is different from $W_1(m_1,m_2,m_3)$.
This is because in our example all $\{\Omega_k\}$ and $B_{11}$ are pure imaginary
and  Jacobi's function  $\vartheta_3(z|B_{11})$ has a $z$-dependent periodic multiplier
$e^{-\pi iB_{11}}e^{-2  \pi i z}$ with respect to $ B_{11}$, i.e.
\[
\vartheta_{3}(z+ B_{11}\mid B_{11})=e^{-\pi i B_{11}}e^{-2 \pi i z}\vartheta_{3}(z\mid B_{11}).
\]
It is the periodic multiplier to give rise to the linear background when $W_2(m_1,m_2,m_3)$
evolves with respect to $\{m_k\}$ via the formula \eqref{W2}.

\section{Concluding remarks}\label{sec-6}

In this paper we constructed algebro-geometric solutions \eqref{eq:4.28} to the lpmKP equation \eqref{eq:1.1}.
The KN spectral problem \eqref{eq:1.3} was employed as an associated spectral problem,
of which the Darboux transformation gives rise to the Lax triad \eqref{eq:2.6} for the lpmKP equation.
Compared with the known one (e.g. Eq.(3.113) in \cite{Hietarinta}), this Lax triad is not explicit,
as the discrete potential function $W$ is defined (via $Z^{(j)}$) by the KN  functions $(u,v)$.
Discrete evolutions are introduced by regarding the Darboux transformation as a map \eqref{eq:3.9}
to generate discrete flows.
The map is shown to be symplectic and integrable, sharing the same integrals
with the continuous Hamiltonian system \eqref{eq:3.1}
which is obtained through the so-called nonlinearisation of the KN spectral problem.
Then, employing algebro-geometric techniques, using  the Baker--Akhiezer functions and Abel map
we introduced the Riemann theta function
and finally obtained the algebro-geometric solutions \eqref{eq:4.28} for the lpmKP equation.
As an example, in Sec.\ref{sec-5}
we presented an explicit solution for the $g=1$ case and illustrated it in Fig.\ref{Fig-1}.

Based on a series of work \cite{CaoX-JPA-2012,CXZ,CZ2012a,CZ2012b,XCN,XCZ-JNMP-2020,XJN},
in Sec.\ref{sec-1}
we have summarized a framework of the approach for constructing algebro-geometric solutions to
multidimensionally consistent systems. The approach has proved effective and  this paper
added one more important successful example.
Reviewing the series of work \cite{CaoX-JPA-2012,CXZ,CZ2012a,CZ2012b,XCN,XCZ-JNMP-2020,XJN} and the present paper,
there are several related problems that are interesting and remain open.
Let us raise them below.

One is to extend solutions to full space.
In fact, in the solutions \eqref{eq:4.28} all $m_i \geq 0$, i.e. the solutions are defined on the first octant.
The same thing happened to the lpKP equation \cite{CXZ}.
For those quadrilateral equations studied in \cite{CaoX-JPA-2012,CZ2012a,CZ2012b,XCN,XCZ-JNMP-2020,XJN},
the obtained algebro-geometric solutions are defined on the first quadrant.
This is because in this approach discrete flows are generated by iterating maps (e.g. \eqref{eq:3.9})
towards one direction.
One may develop the approach to obtain solutions in full space.

The second is to construct algebro-geometric solutions containing two soliton parameters for 3D lattice equations.
Let us explain the problem below.
In general, an integrable 3D equation admits a plane-wave factor (PWF)
with two independent soliton parameters. For example, for the discrete AKP equation, its PWF reads \cite{Miwa-1982}
\[\rho_i=\left(\frac{\beta_1-p_i}{\beta_1-q_i}\right)^{m_1}
\left(\frac{\beta_2-p_i}{\beta_2-q_i}\right)^{m_2}
\left(\frac{\beta_3-p_i}{\beta_3-q_i}\right)^{m_3},\]
where $\beta_1, \beta_2, \beta_3$ are spacing parameters and $p_i$ and $q_i$
are referred to as soliton parameters.
By imposing constraints on $(p_i,q_i)$ one can have a PWF  for   reduced (2D) equations.
Both the lpKP and lpmKP equations allow a PWF with two independent soliton parameters (\S 9.7 of \cite{Hietarinta}),
and for 2D lattice equations, e.g. the ABS equations, their PWF usually reads \cite{HZ-2009,NAH-2009}
\[\rho_i=\left(\frac{\beta_1-p_i}{\beta_1+p_i}\right)^{m_1}
\left(\frac{\beta_2-p_i}{\beta_2+p_i}\right)^{m_2},\]
which contains only one soliton parameter $p_i$.
In the present paper, the lpmKP equation \eqref{eq:1.1} is reconstructed in \eqref{eq:Y123},
as a consequence of 3D consistency of the ldNLS equation \eqref{eq:2.9}.
Note that this does not mean the lpmKP equation \eqref{eq:1.1} is a fake 3D equation,
because it does allows PWF with two  independent soliton parameters (\S 9.7 of \cite{Hietarinta}).
However, this indicates that the solutions we constructed in the paper,
by using 3D consistency of the 2D ldNLS equation,
is a special solution in which  PWF contains a single soliton parameter.
One may develop an approach for 3D lattice equations to construct their algebro-geometric solutions
in which the PWF contains two independent soliton parameters and therefore allows reductions.

The third problem is to apply the scheme to other ABS equations and 3D lattice equations that are 4D consistent
(including octahedron-type equations \cite{ABS-2012} and the discrete BKP and Schwarzian BKP equation \cite{ABS}).
The approach has proved effective and so far H$_1$, H$_3$(0) and Q$_1(\delta)$, lpKP and lpmKP have been solved
with this approach.
However, as we mentioned in Section \ref{sec-1},  for other equations,
the associated  continuous spectral problems are still unknown.
It is also notable that, as we pointed out in Section 1, for H$_1$, H$_3$(0) and Q$_1(0)$,
each equation can have more than one associated  continuous  spectral problems.
For the same equation, these continuous  spectral problems may lead to either same
(for H$_1$ and Q$_1(0)$, cf.\cite{CaoX-JPA-2012,XJN}, and cf. \cite{XCZ-JNMP-2020,XJN})
or  different hyperelliptic curves (for H$_3(0)$, cf.\cite{CZ2012a,XJN}),
but for each equation the obtained solutions have different formulations.

The final problem  is  finite-gap integration based on theory of trigonal curves
for discrete integrable systems.
So far the hyperelliptic curves associated with the two-sheeted Riemann surfaces
were employed in our scheme, cf.\cite{CaoX-JPA-2012,CZ2012a,CZ2012b,XCN,XCZ-JNMP-2020,XJN}.
There were already some exciting developments in the finite-gap integration theory based on trigonal curves
in continuous case, e.g. \cite{DGU-1999,GWH-2011,GWH-2013,GZD-2014}.
It would be very meaningful to develop the theory in discrete case to study
lattice equations related to third-order spectral problem, e.g. the discrete Boussinesq equations \cite{HZ-2021}.

\vspace*{0.2cm}

\subsection*{Acknowledgments}

The authors are grateful to the referees for their invaluable comments.
Our sincere thanks are also extended to Dr. Xing Li
for sharing her expertise of computing path integrations on torus
and providing figures of this paper.
This work is supported by the National Natural Science Foundation of China (grant nos. 11631007, 11875040).

\vskip 20pt

\appendix
\section{The semi-discrete lpmKP and ldNLS equations}\label{A-1}

Since the spectral problem \eqref{eq:2.6a} arises from the Darboux transformation \eqref{eq:1.4},
which commutes with the KN spectral problem and may serve as a Darboux transformation for the whole KN hierarchy,
we can have more equations in this frame,
which compose a hierarchy of the lpmKP family,
in terms of the number of discrete independent variables.

The first semi-discrete lpmKP equation (with two discrete independent variables) is
\begin{equation}\label{eq:2.14}
\Xi^{(1,2)}\equiv (\widetilde{W}-\overline{W})_x+\beta_1^2(e^{-\overline{\widetilde{W}}+\overline{W}}
-e^{-\widetilde{W}+W})-
\beta_2^2(e^{-\widetilde{\overline{W}}+\widetilde{W}}-e^{-\overline{W}+W})=0,
\end{equation}
which is a consequence of the compatibility of \eqref{eq:1.3} \eqref{eq:2.6a} and \eqref{eq:2.6b},
where \eqref{eq:2.8a}, \eqref{eq:2.8b} and \eqref{eq:2.8d} should be used.
Alternatively, equation \eqref{eq:2.14} is obtained from \eqref{eq:2.10} by substituting \eqref{eq:2.8d} into the equation.
Note that  (\ref{eq:2.14}) was also found in \cite{NCW-LNP-1985} (Eq.(4.27))
as a continuum limit of the lpmKP equation \eqref{eq:1.1}.

The second semi-discrete lpmKP equation (with one discrete independent variable),
\begin{equation}\label{eq:2.16}
\Xi^{(2,1)}\equiv(\t{W}+W)_{xx}-(\t{W}_x^2-W_x^2)-2\beta_1^2(e^{-\t{W}+W})_x-(\t{W}-W)_y=0,
\end{equation}
is obtained from the compatibility of \eqref{eq:1.3}, \eqref{eq:2.6a} and the following linear problem
\begin{equation}\label{eq:2.15}
\partial_y\chi=U_2\chi,\quad U_2=\lambda^2U_1
+\left(\begin{array}{cc}
\lambda^2(-uv)&\lambda(u_x-2u^2v)\\
\lambda(-v_x-2uv^2)&-\lambda^2(-uv)
\end{array}\right).
\end{equation}
The calculation is complicated. Let us sketch it below.
First, the compatibility of  (\ref{eq:1.3}) and (\ref{eq:2.15}) yields the dNLS equations
\begin{equation}
\label{eq:2.17a}
\Xi^{(2,0)}_1\equiv u_y-u_{xx}+2(u^2v)_x=0,~\quad \Xi^{(2,0)}_2 \equiv v_y+v_{xx}+2(uv^2)_x=0,
\end{equation}
and the compatibility of  (\ref{eq:1.3}) and (\ref{eq:2.15}) yields relations \eqref{eq:2.2}
 with formulation \eqref{eq:2.3.4},
where we need to replace $Z$ with $Z^{(1)}$.
In particular, \eqref{eq:2.2a} and \eqref{eq:2.2b} indicate that
\begin{equation}\label{eq:AB}
B_1 \doteq (Z^{(1)}u)_x-\t{u}-\beta_1^2u=0,~\quad
B_2 \doteq (Z^{(1)}\t{v})_x+ v+ \beta_1^2\t{v}=0.
\end{equation}
Then, making use of \eqref{eq:2.8d}, \eqref{eq:2.17a} and \eqref{eq:AB},
the compatibility of \eqref{eq:2.6a} and \eqref{eq:2.15} yields
\begin{equation*}
\mathbf{0}=\partial_y D^{(\beta_1)}-\t U_2 D^{(\beta_1)}+D^{(\beta_1)}U_2
=\left(\begin{array}{cc}
-\lambda^2 Z^{(1)} \Xi^{(2,1)} & \lambda \partial^{-1}_x\partial_y B_1 \\
\lambda^3 \partial^{-1}_x\partial_y B_2 & 0
\end{array}\right),
\end{equation*}
which gives rise to equation \eqref{eq:2.16}.
Note that \eqref{eq:2.16} can also be derived in the following alternative way.
From the first equation in \eqref{eq:2.17a} and the shifted second equation
$\t v_y=-\t v_{xx}-2(\t u\t v^2)_x$, one can have
\begin{align}\label{eq:2.17b}
(u\t{v})_y=(u_x\t{v}-u\t{v}_x)_x-2u\t{v}\big[(u_xv+\t{u}\t{v}_x)+(\t{u}\t{v}+uv)_x\big],
\end{align}
which can lead to equation \eqref{eq:2.16} by using \eqref{eq:AB} and
relation $(u\t{v})_y=(\t{W}-W)_y(2-Z^{(1)})/(Z^{(1)})^2$  due to $Z^{(1)}=(Z^{(1)})^2u\t{v}+1$.
We also note that equation (\ref{eq:2.16}) can be converted to its non-potential form
\begin{equation}
-s\Xi^{(2,1)}\equiv s_{xx}+2\beta_1^2ss_x+2[s(\Delta^{-1}\ln s)_x]_x-s_y=0,
\end{equation}
where $s=\ln(W-\t W)$ and $\Delta f=\t f-f$.
This equation was obtained in \cite{Tamizhmani} as a gauge equivalence of a semi-discrete KP equation.

We have derived  the lpmKP equation  $\Xi^{(0,3)}=0$ and two semi-discrete lpmKP
equations $\Xi^{(1,2)}=0$ and $\Xi^{(2,1)}=0$. All these equations have the pmKP
equation (\ref{eq:1.2}), $\Xi^{(3,0)}=0$, as continuum limit.
Let us replace $\beta_k$ by  $\beta_k^{-2}=\varepsilon_k=c_k\varepsilon,\,k=1,2,3$, with non-zero
and distinct constants $c_1,c_2,c_3$,
and consider continuum limits in terms of Miwa's variables
$\vec{t}=(t_1, t_2, t_3, \cdots)$, where \cite{Miwa-1982}
\[t_k=-\frac{1}{k}\sum^{}_{i}\varepsilon_i^k m_i.\]
This indicates that
\[
T_1^{j_1}T_2^{j_2}T_3^{j_3} W(m_1,m_2,m_3)=W\left(x-\sum^3_{i=1}j_i\varepsilon_i,\,
y-\frac{1}{2}\sum^3_{i=1}j_i\varepsilon_i^2,\,
t-\frac{1}{3}\sum^3_{i=1}j_i\varepsilon_i^3\right),
\]
where $x=t_1, y=t_2, t=t_3$ and $j_i\geq 0$ for $i=1,2,3$;
and for differential-difference case,
\[
T_1^{j_1}T_2^{j_2} W(x', m_1,m_2)=W\left(x-\sum^2_{i=1}j_i\varepsilon_i,\,
y-\frac{1}{2}\sum^2_{i=1}j_i\varepsilon_i^2,\,
t-\frac{1}{3}\sum^2_{i=1}j_i\varepsilon_i^3\right),
\]
where $x=x'+t_1, y=t_2, t=t_3$ and $j_i\geq 0$ for $i=1,2$;
\[
T_1^{j_1} W(x', y', m_1)=W\left(x-j_1\varepsilon_1,\,
y-j_1\frac{\varepsilon_1^2 }{2},\,
t-j_1\frac{\varepsilon_1^3 }{3}\right),
\]
where $x=x'+t_1, y=y'+t_2, t=t_3$ and $j_1\geq 0$.
It turns out that
%
\begin{align*}
&\Xi^{(2,1)}(x', y', m_1)=\Xi^{(3,0)}\frac{2}{3}c_1^2\varepsilon^2+O(\varepsilon^3),\\
&\Xi^{(1,2)}(x', m_1,m_2)=\Xi^{(3,0)}\frac{1}{3}c_1c_2(c_1-c_2)\varepsilon^3+O(\varepsilon^4),\\
&\Xi^{(0,3)}(m_1,m_2,m_3)
=\Xi^{(3,0)}\frac{1}{3}\big(c_1c_2(c_1-c_2)+c_2c_3(c_2-c_3)+c_3c_1(c_3-c_1)\big)\varepsilon^3
+O(\varepsilon^4).
\end{align*}

Similarly, equations \eqref{eq:2.5} and (\ref{eq:2.9}), up to some transformations,
yield the dNLS equations (\ref{eq:2.17a}) in continuum limit.
In fact, for equation \eqref{eq:2.5}, we first replace $u$ and $v$ by
$(-\beta_1^2)^{m_1}u$, and $(-\beta_1^2)^{-m_1}v$, respectively,
and rewrite \eqref{eq:2.5} as
\begin{align*}
&\Xi_1^{'(1,1)}(x,m_1)\equiv u_x+(\t{u} \t{v} -u  v)u+\frac{1}{2}\Big(1+\sqrt{1+4\beta_1^{-2}u\t{v}}\,\Big)
\beta_1^2(\t{u}-u)=0,\\
&\Xi_2^{'(1,1)}(x,m_1)\equiv \t{v}_x+(\t{u}\t{v}-uv)\t{v}+\frac{1}{2}\Big(1+\sqrt{1+4\beta_1^{-2}u\t{v}}\,\Big)
\beta_1^2(\t{v}-v)=0,
\end{align*}
which is available for taking continuum limit.
Then, let $\beta_1^{-2}=\varepsilon_1$ and define
\[T_1^j{f}(x',m_1)=f(x-j\varepsilon_1,\,y-j\varepsilon_1^2/2),\]
where $x=x'+t_1, y=t_2$ and $j\geq 0$.
Then we have $(\varepsilon_1\sim 0)$
\begin{align*}
& \Xi_1^{'(1,1)}(x',m_1)=-\frac{\varepsilon_1}{2} \Xi_1^{(2,0)}+O(\varepsilon_1^2),\\
& \Xi_2^{'(1,1)}(x',m_1)=-\frac{\varepsilon_1}{2} \Xi_2^{(2,0)}+O(\varepsilon_1^2),
\end{align*}
where $\Xi_i^{(2,0)}$ are given in \eqref{eq:2.17a}.
Finally, in \eqref{eq:2.9}, replace $u$ by $(-\beta_1^2)^{m_1}(-\beta_2^2)^{m_2}u$ and $v$ by
$(-\beta_1^2)^{-m_1}(-\beta_2^2)^{-m_2}v$, and rewrite the equation as
\begin{align*}
&\Xi_1^{'(0,2)}\equiv \beta_1^2(\t{Z}^{(2)} \t{u}-Z^{(2)} u)-\beta_2^2(\b{Z}^{(1)} \b{u}-Z^{(1)} u)=0,\\
&\Xi_2^{'(0,2)}\equiv \beta_1^2(\b{Z}^{(2)}\t{\b{v}}-Z^{(2)}\b{v})-\beta_2^2(\b Z^{(1)}\b{\t{v}}
-Z^{(1)}\t{v})=0.
\end{align*}
Let $\beta_i^{-2}=\varepsilon_i= c_i\varepsilon,$ and introduce
\[
T_1^{j_1}T_2^{j_2} f(m_1,m_2)=f \left(x-\sum^2_{i=1}j_i\varepsilon_i,\,
y-\frac{1}{2}\sum^2_{i=1}j_i\varepsilon_i^2\right),
\]
where $x=t_1, y=t_2$ and $j_i\geq 0$ for $i=1,2$.
The continuum limit yields
\begin{align*}
& \Xi_1^{'(0,2)}=-\frac{(c_1-c_2)\varepsilon}{2} \Xi_1^{(2,0)}+O(\varepsilon^2),\\
& \Xi_2^{'(0,2)}=-\frac{(c_1-c_2)\varepsilon}{2} \Xi_2^{(2,0)}+O(\varepsilon^2).
\end{align*}

In the rest part of this section, we review some links between $(1+1)$-dimensional and
$(2+1)$-dimensional integrable systems,
which maybe helpful for understanding the connections in the discrete case.
Let us start with the continuous mKP hierarchy, which are generated from the compatibility of (see \cite{Konop-1992})
\begin{subequations}\label{mkp-lax}
\begin{align}
&L\varphi = \lambda \varphi,~~ L=\partial_x + w_0 + w_1 \partial_x^{-1}+ w_2 \partial_x^{-2}+\cdots,
\label{mkp-lax-L}\\
&\varphi_{t_j}=A_j \varphi,
\label{mkp-lax-tj}
\end{align}
\end{subequations}
where $w_j=w_j(x=t_1, y=t_2, t=t_3, t_4, \cdots)$,
$\partial_x=\partial/\partial x$, $A_j=(L^j)_{\geq 1}$ stands for the pure differential part of $L^j$.
One can express $w_j$ for $j\geq 1$ in terms of $w_0\equiv w$ using the compatibility between \eqref{mkp-lax-L}
and $\varphi_y=A_2 \varphi$,
and then the mKP hierarchy $w_{t_j}=P_j(w)$ arise from the  compatibility  of $\varphi_y=A_2 \varphi$ and
\eqref{mkp-lax-tj}.
By $A_j^*$ we denote the operator adjoint of $A_j$.
Introduce
\begin{equation}\label{mkp-lax-tj-ad}
\psi_{t_j}=-A_j^* \psi.
\end{equation}
It can be verified that for the above eigenfunctions $\varphi$ and $\psi$,
the mKP hierarchy admits a symmetry $\sigma=(\varphi\psi)_x$.
Consider the symmetry constraint $w=\varphi\psi$.
In the following we denote $\varphi=u$ and $\psi=v$,
thus we have (cf. $w=W_x=uv$ in \eqref{eq:2.8d})
\begin{equation}
w=uv.
\end{equation}
With this constraint, the spectral problem \eqref{mkp-lax-L} is converted to the KN spectral problem \eqref{eq:1.3}
(up to gauge transformation),
and the coupled system \eqref{mkp-lax-tj} and \eqref{mkp-lax-tj-ad} (with $\varphi=u, ~\psi=v$)
give rise to the KN hierarchy (up to gauge transformation)
\begin{equation}\label{KN-hie}
\mathbf{u}_{t_j}=K_j(\mathbf{u}),~~ \mathbf{u}=(u,v)^T.
\end{equation}
For more details, one may refer to \cite{Chen-2002}.
Note that the KN hierarchy \eqref{KN-hie} also result  from the compatibility of the KN spectral problem \eqref{eq:1.3}
and the time evolution $\chi_{t_j}=U_j\chi$.
In this context, it is not surprised that the potential mKP equation \eqref{eq:1.2}
can be obtained from the compatibility of
the triad $\chi_{t_j}=U_j\chi$ for $j=1,2,3$ and with $t_1=x, t_2=y, t_3=t$ and $W_x=uv$.
Such a link between $(1+1)$-dimensional and
$(2+1)$-dimensional integrable systems
was first revealed in \cite{Cao1990,ChengL-PLA-1991,KSS-PLA-1991} in the early 1990s
for the ZS-AKNS and KP systems.
For the differential-difference case with one discrete independent variable,
the analogue results build connections between the semi-discrete ZS-AKNS
and differential-difference KP systems (see \cite{ChenDZ-2017}).
The obtained semi-discrete ZS-AKNS spectral problem is nothing but the Darboux transformation
of the ZS-AKNS spectral problem \eqref{AKNS-sp}, and it has been used to derive algebro-geometry
solutions for the lpKP equation \cite{CXZ}.
However, for the differential-difference mKP hierarchy, its eigenfunction symmetry constraint
gives rise the relativistic Toda spectral problem (see \cite{ChenZZ-2021}),
which is not the Darboux transformation \eqref{eq:1.4} that we used in this paper.
Anyway, as we can see, in our approach it is an important step to establish the link between the
lpmKP and KN equations.
The eigenfunction symmetry constraint can provide some insights but not always as direct as expected.

\section{Riemann theta function associated with hyperelliptic curve}\label{B-1}

In this section we introduce how a Riemann theta function arises from a hyperelliptic curve.
For more details one can refer to \cite{Farkas,Griffiths,Mumford}.
This section also provides some complex algebraic geometry preliminaries that will be used in our approach.

For a generic hyperelliptic curve $\mathcal{R}$:  $\xi^2=R(\zeta)$ with genus $g$ (for example, \eqref{R-curve}),
the following Abelian differentials of the first kind constitute the basis of holomorphic differentials of $\mathcal R$:
\begin{equation}\label{eq:B1}
\tilde{\omega}_{j}=\displaystyle\frac{\zeta^{g-j}\mathrm{d}\zeta}{2\sqrt{R(\zeta)}},\ \ j=1,\ldots,g.
\end{equation}
Let the closed loops $a_{1},\cdots,a_{g},b_{1},\cdots,b_{g}$ be the canonical basis of the homology
group $\mathrm{H}_{1}(\mathcal{R})$,
with intersection numbers
\[a_j \circ b_k = \delta_{jk},~~ a_j \circ a_k = 0,~~  b_j \circ b_k = 0,~~ j, k = 1,2,\cdots, g.\]
Note that here we adopt the conventional notations $\{a_j, b_j\}$ for these canonical basis,
without making confusion with the discrete potentials $(a_m, b_m)$ that we used in Sec.\ref{sec-4}.
The basis of holomorphic differentials can be  normalized as follows:
\begin{align}\label{eq:B2}
\begin{split}
\vec{\omega}\doteq (\omega_1,\cdots,\omega_g)^T
&=C(\tilde{\omega}_{1},\cdots,\tilde{\omega}_{g})^T\\
&=(C_{1}\zeta^{g-1}+C_{2}\zeta^{g-2}+\ldots
C_{g})\displaystyle\frac{\mathrm{d}\zeta}{2\sqrt{R(\zeta)}},
\end{split}
\end{align}
where
\begin{equation}\label{eq:B3}
C=(A_{jk})^{-1}_{g\times g},\ \ A_{jk}=\int_{a_{k}}\tilde{\omega}_{j},
\end{equation}
and  $C_{l}$ stands for  the \emph{l}-th column vector of $C$.
For this normalized basis $\vec{\omega}$, we have
\begin{equation}\label{eq:B4}
\int_{a_{k}}\omega_{j}=\delta_{jk},\ \ \int_{b_{k}}\omega_{j}=B_{jk},
\end{equation}
where the matrix $B=(B_{jk})_{g\times g}$ is symmetric with positive definite imaginary part.
$B$ can be used as a  periodic matrix to defined the Riemann theta function
(which is holomorphic):
\begin{equation}\label{eq:B5}
\theta(z;B)=\sum_{z^{\prime}\in \mathbb{Z}^{g}}\exp\pi\sqrt{-1}(<Bz^{\prime},z^{\prime}>+2<z,z^{\prime}>),\ \ z\in \mathbb{C}^{g},
\end{equation}
where by $<\cdot,\cdot>$ we denote the scalar product in $\mathbb{C}^{g}$.

Next, let us reveal features of zeros of the above Riemann theta function.
Let $\mathcal T$ be a period lattice $\mathcal T=\{z\in \mathbb{C}^{g}\mid z
=\vec{m}+B\vec{n},~ \vec{m},\vec{n}\in \mathbb{Z}^{g}$\}.
The quotient space $J(\mathcal R)=\mathbb C^g/\mathcal T$ is called the Jacobian variety of $\mathcal R$.
Introduce the Abel map (also called Abel-Jacobi map) $\mathcal A:\mathcal R\rightarrow J(\mathcal R)$,  by
\begin{equation}\label{eq:B6}
\mathcal{A}(\mathfrak p)=\int_{\mathfrak{p}_0}^{\mathfrak{p}}\vec{\omega}
=\left(\int_{\mathfrak{p}_0}^{\mathfrak{p}}\omega_1, \int_{\mathfrak{p}_0}^{\mathfrak{p}}\omega_2,
~\cdots,\int_{\mathfrak{p}_0}^{\mathfrak{p}}\omega_g\right)^T,
\end{equation}
where $\mathfrak{p}, \mathfrak{p}_0 \in \mathcal R$ with $\mathfrak{p}_0$ being some choosing fixed point.
The map can be extended to the divisors
by defining
\begin{equation}\label{eq:B8}
 \mathcal{A}\left (\sum_k n_{k}\mathfrak{p}_k\right )=\sum_k n_{k}\mathcal{A}(\mathfrak{p}_k).
\end{equation}
Consider the function
\begin{equation}\label{eq:B9}
\theta(z(\mathfrak{p},D);B)
=\theta(-\mathcal A(\mathfrak p) +\mathcal A(D)+\vec K);B), ~ \ \ \mathfrak p\in \mathcal R,
\end{equation}
where the divisor $D=\sum_{k=1}^{g}\mathfrak p_{k}$
with $\{\mathfrak p_{k}\}$ being $g$ distinct points on $\mathcal{R}$,
and the vector $\vec K$ of Riemann constants is defined as
\begin{equation}\label{eq:B10}
\vec K=-\sum_{k=1}^{g}\bigg[\int_{a_{k}}\mathcal A\,\omega_{k}
-\bigg(\displaystyle\frac{B_{kk}}{2}+\mathcal A_{k}(\mathfrak p^{\prime})\bigg)
\int_{a_{k}}\vec{\omega}\bigg],
\end{equation}
where $\mathcal A_{k}$ is the $k$-th component of the Abel map $\mathcal A$,
and $\mathfrak p^{\prime}$
is the base point of the fundamental group $\pi_{1}(\mathcal R)$
(see \cite{Farkas}).
By choosing $\mathfrak p^{\prime}=\mathfrak{p}_0$, the above formula reduces to
\begin{equation}\label{eq:B11}
\vec K=-\sum_{k=1}^{g}\bigg(\int_{a_{k}}\mathcal A\omega_{k}-\displaystyle\frac{B_{kk}}{2}\int_{a_{k}}\vec{\omega}\bigg).
\end{equation}
 According to the Riemann vanishing theorem,
 the zeros of $\theta(z(\mathfrak{p},D),B)$ are $\mathfrak p_{1},\ldots,\mathfrak p_{g}$.
Due to this fact, the Riemann theta function \eqref{eq:B9}  is often used to characterize
meromorphic functions with certain  zeros and poles.

\vskip 30pt


\end{document}